\documentclass[acmsmall]{acmart}
\AtBeginDocument{%
  }

\usepackage{array}
\newcolumntype{C}[1]{>{\centering\arraybackslash}m{#1}}

\usepackage[ruled, linesnumbered]{algorithm2e}
\usepackage{changepage}
\usepackage{graphicx}
\usepackage{multirow}
\usepackage{tabularx}
\usepackage{subcaption}
\usepackage{svg}
\usepackage{colortbl}
\usepackage{fancyhdr}
\usepackage{xcolor}
\usepackage{caption}
\usepackage{tablefootnote}
\usepackage{lineno}
\usepackage[normalem]{ulem} 
\usepackage{makecell}       

\newcommand{\revise}[1]{\textcolor{black}{#1}}
\newcommand{\reviseagain}[1]{\textcolor{black}{#1}}
\newcommand{\red}[1]{\textcolor{black}{#1}}
\newcommand{\blue}[1]{\textcolor{black}{#1}}
\newcommand{\orange}[1]{\textcolor{black}{#1}}

\setcopyright{cc}
\setcctype{by}
\acmJournal{PACMMOD}
\acmYear{2025} \acmVolume{3} \acmNumber{6 (SIGMOD)} \acmArticle{343} \acmMonth{12} \acmPrice{}\acmDOI{10.1145/3769808}

\begin{document}

\title{N2E: A General Framework to Reduce Node-Differential Privacy to Edge-Differential Privacy for Graph Analytics}

\author{Yihua Hu}
\affiliation{%
  \institution{Nanyang Technological University}
  \city{Singapore}
  \country{Singapore}
}
\email{yihua001@e.ntu.edu.sg}

\author{Hao Ding}
\affiliation{%
  \institution{Nanyang Technological University}
  \city{Singapore}
  \country{Singapore}
}
\email{HAO028@e.ntu.edu.sg}

\author{Wei Dong}
\authornote{Wei Dong is the corresponding author.}
\affiliation{%
  \institution{Nanyang Technological University}
  \city{Singapore}
  \country{Singapore}
}
\email{wei_dong@ntu.edu.sg}

\renewcommand{\shortauthors}{Yihua Hu, Hao Ding, and Wei Dong}

\begin{abstract}
Differential privacy (DP) has been widely adopted to protect sensitive information in graph analytics. While edge-DP, which protects privacy at the edge level, has been extensively studied, node-DP, offering stronger protection for entire nodes and their incident edges, remains largely underexplored due to its technical challenges.
A natural way to bridge this gap is to develop a general framework for reducing node-DP graph analytical tasks to edge-DP ones, enabling the reuse of existing edge-DP mechanisms.
A straightforward solution based on group privacy divides the privacy budget by a given degree upper bound, but this leads to poor utility when the bound is set conservatively large to accommodate worst-case inputs.
To address this, we propose node-to-edge (N2E), a general framework that reduces any node-DP graph analytical task to an edge-DP one, with the error dependency on the graph’s true maximum degree. 
N2E incorporates two \revise{main} techniques: a distance-preserving clipping mechanism that bounds edge distance between neighboring graphs after clipping, and the first node-DP mechanism for maximum degree approximation, enabling tight, privacy-preserving clipping thresholds.
By instantiating N2E with existing edge-DP mechanisms, we obtain the first node-DP solutions for tasks such as maximum degree estimation. 
For edge counting, our method theoretically matches the error of the state-of-the-art, which is provably optimal, and significantly outperforms existing approaches for degree distribution estimation. Experimental results demonstrate that our framework achieves up to a $2.5\times$ reduction in error for edge counting and up to an $80\times$ reduction for degree distribution estimation.
\end{abstract}



\received{April 2025}
\received[revised]{July 2025}
\received[accepted]{August 2025}

\maketitle

\section{Introduction}
\label{sec:intro}

In the era of big data, graph analytics has become essential across various domains such as social networks, healthcare, and finance, yielding significant societal and economic benefits.
For instance, social platforms such as Facebook and LinkedIn leverage graph analytics to support features like friend and professional connection recommendations. 
However, these graphs often contain sensitive information, and publishing analytical results can pose serious privacy risks, such as revealing specific edges or exposing node attributes~\cite{narayanan2009anonymizing, wu2021adapting, zhang2022inference}. 
To address these concerns, \textit{differential privacy} (DP)~\cite{dwork2006calibrating} has emerged as a rigorous standard for privacy-preserving graph analytics. 
Under DP, two graphs are considered neighbors if they differ by a single individual, represented either as a node or an edge.
DP ensures that the analytical results between two neighboring graphs are $(\varepsilon,\delta)$-indistinguishable, preventing the inference of any individual's information.
The privacy budget $\varepsilon$ controls privacy leakage, with smaller values providing stronger privacy guarantees, while $\delta$ controls the failure probability, quantifying the likelihood that the privacy guarantee does not hold.

A key distinction in applying DP to graph analytics lies in the choice of protecting either edges or nodes. 
Specifically, 
\emph{edge-DP}~\cite{hay2009accurate, karwa2011private, sala2011sharing} protects individual edges, while \emph{node-DP}~\cite{kasiviswanathan2013analyzing, chen2013recursive, day2016publishing} protects entire nodes with its incident edges.
Accordingly, two graphs are neighbors under edge-DP if they differ by a single edge, whereas under node-DP, they differ by a node and all its edges.   
Although node-DP offers stronger privacy guarantees, it is more challenging to achieve since a node has a greater contribution to the analytical result than a single edge.
This leads to a significant research gap between node-DP and edge-DP.
Many graph analytical tasks, such as maximum degree estimation~\cite{hay2009accurate, karwa2011private,dong2024continual}, and shortest path computation~\cite{sealfon2016shortest,fan2022private} have been extensively studied under edge-DP but lack corresponding solutions under node-DP.
Yet, node-DP is more meaningful in practice, as it better aligns with real-world privacy concerns, 
where users (nodes) and their associated relationships (edges) should be protected as a whole rather than in isolation.

Therefore, a well-motivated objective is to develop a general framework for reducing node-DP tasks to edge-DP. 
More precisely, given any graph task $Q$ and an edge-DP mechanism $\mathcal{M}_Q^{\mathrm{edge}}$, this framework produces a corresponding node-DP protocol $\mathcal{M}_Q^{\mathrm{node}}$.  
Recall that the key challenge in achieving node-DP, as opposed to edge-DP, lies in the greater impact of a single node on the graph structure.
However, this impact can be naturally bounded, as each node is incident to at most $O(\widehat{N})$~\footnote{\red{We use $O(\cdot)$ for standard asymptotic bounds and $\tilde{O}(\cdot)$ to suppress logarithmic factors. 
All logarithms are written as $\log$ and treated equivalently in $O(\cdot)$.}}
edges, where $\widehat{N}$ is a predefined upper bound on the number of nodes $N$ in the graph. 
Consequently, protecting an edge's privacy also protects part of the associated node's privacy.
This principle is formally captured in the DP literature as \textit{group privacy}~\cite{dwork2014algorithmic} (see Lemma~\ref{lem:gp} for a formal definition).
Intuitively, since two neighboring graphs under node-DP have an \textit{edge distance}, which is defined as the number of differing edges, bounded by $O(\widehat{N})$, any $(\varepsilon/\widehat{N},\delta/\widehat{N})$-edge-DP protocol naturally follows $(\varepsilon,\delta)$-node-DP.

However, this approach always breaks utility in practice. 
Since most edge-DP mechanisms exhibit errors that grow linearly with ${1}/{\varepsilon}$, the group privacy solution amplifies the error by a factor of $\widehat{N}$. 
In practice, $\widehat{N}$ should be set conservatively large to ensure coverage of all possible graphs, leading to a large degradation in utility.
Take edge counting as an example, applying group privacy with a state-of-the-art edge-DP protocol results in an error of $O({\widehat{N}}/{\varepsilon})$.
Instead, the state-of-the-art node-DP solution achieves an error of $\tilde{O}({\deg(G)}/{\varepsilon})$~\cite{dong2022r2t,fang2022shifted}, where $\deg(G)$ is the maximum degree of nodes in the graph $G$.
In real-world scenarios, $\deg(G)$ is typically much smaller than $\widehat{N}$.
For example, the average Facebook user has around $130$ friends, far fewer than the total number of users of the platform.

The bad performance of the group privacy approach stems from its dependency on $\widehat{N}$, the worst-case upper bound on a node’s contribution. 
A more preferable solution is to use the actual maximum node's contribution in the graph $G$, i.e., $\deg(G)$. 
However, $\mathrm{deg}(G)$ itself is highly sensitive, and using it as a parameter to scale randomness
with group privacy would break privacy guarantees~\cite{huang21mean,dong2022r2t,dong2023universal}.
A common strategy to mitigate worst-case error dependency is the \textit{clipping mechanism}~\cite{blocki2013differentially, huang21mean}.  
By clipping the input graph with a degree upper bound $\tau$, 
the resulting node-DP error shifts its dependency from $\widehat{N}$ to $\tau$.
Moreover, when $\tau$ is set close to $\deg(G)$, we achieve the error dependency that improves from $\widehat{N}$ to $\deg(G)$.
For example, in the node-DP edge counting problem, 
we can achieve the state-of-the-art error $O(\deg(G)/\varepsilon)$. 
However, this approach introduces two key challenges.

\textbf{Challenge 1. Distance-preserving clipping mechanism.}
The clipping mechanism needs to be distance-preserving, meaning that for a given degree upper bound $\tau$,
the edge distance of two neighboring graphs after clipping remains upper bounded by a function of $\tau$.  
However, no existing clipping mechanism satisfies this requirement under node-DP.
To illustrate this, consider the most straightforward clipping mechanism that deletes all nodes with degrees higher than $\tau$. 
Such a clipping mechanism fails to meet the distance-preserving requirement under node-DP.
Consider a graph $G$ where all $N$ nodes have degree $\tau$, and construct a neighboring graph $G'$ by adding a new node that connects to all $N$ nodes in $G$. 
After applying the clipping mechanism, $G$ remains unchanged, but $G'$ becomes empty.
This results in an edge distance of ${N \cdot \tau}/{2}$, which violates the distance-preserving requirement.

\textbf{Challenge 2. Node-DP maximum degree estimation.}
To achieve an error dependency on $\deg(G)$, we should estimate a $\tau$ as close to $\deg(G)$ as possible.
\reviseagain{Existing works~\cite{raskhodnikova2016lipschitz,liu2019private} primarily address this problem by formulating it as the parameter selection for specific downstream tasks, whereas we aim to develop a reusable, task-agnostic solution for maximum degree estimation with explicit utility guarantees.}
First, extending edge-DP maximum degree estimators to support node-DP is fundamentally infeasible.
Under edge-DP, each edge affects the maximum degree by at most $1$, so adding $O(1/\varepsilon)$ noise is sufficient to mask its contribution.
In contrast, under node-DP, a single node can influence the maximum degree by up to $\widehat{N}$, requiring noise of scale $O(\widehat{N} / \varepsilon)$ to ensure privacy. However, this completely destroys utility, as the maximum degree is at most $\widehat{N}$.
This challenge can be further reflected in the fact that many node-DP techniques~\cite{kasiviswanathan2013analyzing,blocki2013differentially,day2016publishing,sajadmanesh2023gap,zhang2020community} assume a predefined degree upper bound rather than proposing a method to estimate it.

\textbf{Problem Statement.~}
These challenges raise a question:  \textit{Is there a general framework to reduce node-DP tasks to edge-DP with an error upgraded at most by $\tilde{O}(\deg(G))$?}  


\subsection{Our Results}
\label{subsec:intro-results}
This paper answers this question affirmatively.  
First, to address \textbf{Challenge 1}, 
we \revise{apply an established clipping mechanism~\cite{blocki2013differentially}, which has been utilized for edge-DP social network analysis.}
\revise{We then show that it satisfies the required privacy-preserving property:}
given a degree upper bound $\tau$, the edge distance of any two neighboring graphs after clipping is bounded by $\tau + k$, where $k$ is the number of nodes in $G$ with degree at least $\tau$. 
When $\tau$ approaches $\deg(G)$, $k$ remains small, satisfying the distance-preserving requirement.
Furthermore, the clipping mechanism \revise{we use} is computationally efficient, operating in linear time.
\reviseagain{
Notably, the same clipping mechanism with the same property has been established in an earlier work~\cite{jain2024time} that studies a reduction from node-DP to edge-DP in the dynamic setting given a fixed degree threshold, of which we were unaware when writing this paper. 
}
\revise{
A detailed comparison between their work and ours is provided at the end of the related work. 
Importantly, the main technical contribution of this paper is our solution to \textbf{Challenge 2}, which, to the best of our knowledge, is the first node-DP mechanism 
\reviseagain{designed specifically}
for maximum degree approximation and is the critical ingredient for completing our general framework.
}

Next, to address \textbf{Challenge 2}, 
we propose our node-DP protocol for maximum degree approximation.
For any input graph $G$, the mechanism outputs a $\tau^*$, which can be shown to be a good approximation of $\mathrm{deg}(G)$.
On one hand, $\tau^*$ is not much smaller than $\mathrm{deg}(G)$ with a bounded rank error: there are at most $\tilde{O}(1 / \varepsilon)$ nodes\footnote{All results in the introduction hold with constant probability.} in $G$ with degree at least $\tau^*$.
On the other hand, $\tau^*$ is at most $O(\deg(G)) + \tilde{O}(1/ \varepsilon)$.
Our node-DP maximum degree approximator relies on solving $O(\log\mathrm{deg}(G))$ linear programs, resulting in polynomial-time complexity.

\begin{table}[htbp]
\centering
\resizebox{1\linewidth}{!}{
    \renewcommand{\arraystretch}{1.2}
    \setlength{\tabcolsep}{12pt}
    \begin{tabular}{c||c|c|c}
        \toprule
        \textbf{Task} & \textbf{Our Results} & \textbf{SOTA Edge-DP} & \textbf{SOTA Node-DP} \\
        \midrule\midrule
        Edge Counting & 
        $O\Bigl(
        \deg(G)\log\log\deg(G) + 
        \log\tfrac{1}{\delta}
        \Bigr)$ &
        $O\Bigl(1
        \Bigr)$ &
        \makecell[c]{
$O\bigl(\deg(G)\log\log\deg(G)\bigr)$~\cite{dong2024instance} \\ 
\reviseagain{$\operatorname{err}(\hat D)
= \operatorname{err}(D^*) \cdot 
O\Bigl(
\log\log n
\Bigr)$~\cite{raskhodnikova2016lipschitz}
} 
}
         \\
        \midrule
        \begin{tabular}{@{}c@{}}Maximum Degree \\ Estimation\end{tabular} &
        \begin{tabular}{@{}c@{}}$|\deg(G) - \deg({G}')|,\,$ 
        $d_{\text{edge}}(G, {G}') = $ \\ 
        $O\Bigl(
        \deg(G)\log\deg(G) 
        + 
        \log\deg(G) 
        \log\tfrac{1}{\delta}
        \Bigr)$\end{tabular}
        &
        $O\Bigl(1\Bigr)$ &
        $O\Bigl({\widehat{N}}\Bigr)$ \\
        \midrule
        Degree Distribution &
        $O\Bigl(
        \deg(G)^{1.5} 
        +
        \mathrm{deg}^{0.5}(G)\log\frac{1}{\delta}
        \Bigr)$ &
        $O\Bigl(
        \deg(G)^{0.5}\Bigr)$ &
        \makecell[c]{
\reviseagain{$\mathbb E[$\text{err}$]  = O\Bigl(\bar{d}(G)^{\frac{2\alpha}{\alpha+1}}
  n^{\frac{2}{\alpha+1}}\Bigr)$~\cite{raskhodnikova2016lipschitz}} \\
    $O\Bigl(
\widehat{N}^{1.5}\Bigr)~$\cite{day2016publishing} 
    }
   
    \\
        \bottomrule
    \end{tabular}
}

\caption{Summary of node-DP error on common graph analytical tasks.\protect\footnotemark 
Here, SOTA stands for state-of-the-art.
\reviseagain{We use $d_{\text{edge}}(G,G')$ to denote the edge distance between $G$ and $G'$, $\bar d(G)$ to denote the average degree, and $\alpha$ to denote the decay parameter of $G$, i.e., $G$ satisfies $\alpha$-decay.
Moreover, $\operatorname{err}(D)$ denotes the bias-plus-noise proxy for the expected absolute error at candidate degree threshold $D$. The generalized exponential mechanism selects
$\hat D$, whereas
$D^*$ denotes the optimal
threshold.
See~\cite{kasiviswanathan2013analyzing,
raskhodnikova2016lipschitz} for the detailed definition of $\operatorname{err}(D)$.
For degree distribution, we use $\ell_1$ error as the error metric.
}
\cite{day2016publishing} also proposes a mechanism that relies on a heuristic solution to estimate $\mathrm{deg}(G)$, but it lacks formal utility guarantees.
We include this heuristic as a baseline in our experiments.
$\varepsilon$ is regarded as a constant and hidden in the $O$-notation.
}
\label{tab:node_dp_errors}
\end{table}

By combining these two mechanisms, we develop our framework \textit{node-to-edge} (N2E), a framework that can reduce an arbitrary node-DP task to an edge-DP task. 
N2E works as follows. 
Given any graph analytical task $Q$ and an edge-DP mechanism $\mathcal{M}_Q^{\mathrm{edge}}$, for any graph $G$, we first compute $\tau^*$ as the approximation of $\mathrm{deg}(G)$.
Then, we clip the graph using our distance-preserving clipping mechanism with degree upper bound $\tau^*$ to obtain the clipped graph $\overline{G}$. 
Finally, this framework invokes the edge-DP mechanism $\mathcal{M}_Q^{\mathrm{edge}}$ on the clipped graph, with privacy budgets $\varepsilon'=\varepsilon/O(\tau^*)$ and $\delta'=\delta/O(\tau^*)$.

\footnotetext{
\red{
Our results satisfy $(\varepsilon, \delta)$-node-DP, while the node-DP SOTA methods satisfy pure $\varepsilon$-node-DP. 
We set $\delta$ to a negligible value of $2^{-30}$ to ensure minimal practical impact. 
Importantly, our N2E framework supports transforming arbitrary node-DP tasks into edge-DP
ones, which the SOTA methods cannot achieve. 
}
}

N2E achieves an error of:
\begin{equation*}
    |Q(G) - Q(\overline{G})| + \mathrm{Err}_Q^{\text{edge}}(\overline{G},\varepsilon',\delta') ,
\end{equation*}
where $\overline{G}$ is obtained by clipping edges of $O(\log \log \deg(G) / \varepsilon)$ nodes, 
and
$\mathrm{Err}_Q^{\text{edge}}$ denotes the error of the edge-DP mechanism on the task $Q$.  
Such an error consists of two components.
The first component is the clipping bias, which arises from the clipping process. This component can be further shown to be optimal (see Section~\ref{subsec:method-combine} for details).
The second component results from the edge-DP mechanism. Since $\overline{G}$ is derived by clipping only a small number of nodes from $G$, most edge-DP protocols exhibit similar errors on both $G$ and $\overline{G}$ under the same parameters. Moreover, since $\varepsilon$ is usually regarded as a constant in the DP literature, the second component aligns with our target error bound: $\mathrm{Err}_Q^{\text{edge}}(G,\varepsilon/\tilde{O}(\mathrm{deg}(G)),\delta/\tilde{O}(\mathrm{deg}(G)))$.
We will show our results upgrade the edge-DP error by $\tilde{O}(\mathrm{deg}(G))$ for the most common graph analytical tasks in Table~\ref{tab:node_dp_errors}.
For edge counting, our result matches the state-of-the-art node-DP error, up to an additional additive term of $O(\log(1/\delta))$.
For degree distribution, our result reduces the error from $O(\widehat{N}^{1.5})$
to 
${O}(\deg(G)^{1.5})+\tilde{O}(\deg(G)^{0.5})$.
For maximum degree estimation, we propose the first node-DP solution: we return the maximum degree of $G'$, where $G'$ deletes $\tilde{O}(\deg(G))$ edges from $G$.

Experimental results~\footnote{Code is available at https://github.com/Chronomia/N2E.} show that N2E achieves up to a $2.5\times$ error reduction compared to 
state-of-the-art methods in edge counting, 
and up to a \red{$80\times$} reduction compared to the existing state-of-the-art method in degree distribution.
Moreover, N2E delivers high utility as the first node-DP solution for maximum degree estimation, achieving a relative rank error of around $1\%$ in most experiments.

\section{Related Work}
\label{sec:related_work}

DP graph analytics has been extensively studied, with early works focusing primarily on graph pattern counting queries. Most of these studies adopt edge-DP and target specific graph patterns. \cite{nissim2007smooth} proposed \textit{smooth sensitivity}, which enables private triangle counting. \cite{karwa2011private} extended this approach to $k$-star counting and introduced \textit{higher-order local sensitivity} for $k$-triangle counting. \cite{zhang2015private} proposed the ladder function, applicable to $k$-triangle and $k$-clique counting. \cite{johnson2018towards} introduced \textit{elastic sensitivity}, the first edge-DP mechanism supporting arbitrary graph pattern counting. More recently, \cite{dong21:residual} proposed \textit{residual sensitivity}, which also supports arbitrary patterns and has been shown to achieve near-optimal error~\cite{dong2021nearly}. 
Besides, some works study edge-DP graph pattern counting in dynamic settings~\cite{fichtenberger2021differentially,dong2024continual} or decentralized settings~\cite{Imola2021,Imola2022,He2025,eden2023triangle}, which is different from our setting.
Moving towards node-DP, \cite{kasiviswanathan2013analyzing} and \cite{blocki2013differentially} proposed solutions for arbitrary graph pattern counting, albeit under strong assumptions on node degrees. \cite{chen2013recursive} introduced the recursive mechanism, which achieves high utility but suffers from high computational cost. More recently, \cite{dong2022r2t,dong2024instance} proposed efficient mechanisms that achieve both high utility and computational efficiency under node-DP. \red{\cite{dong2023better} studied how to answering multiple queries.}
Besides graph pattern counting, another widely studied problem is degree distribution.
In this context, several works have investigated edge-DP~\cite{cormode2012differentially, acs2012differentially, qardaji2013understanding, zhang2014towards}, while~\reviseagain{\cite{day2016publishing,jain2024time, kasiviswanathan2013analyzing, raskhodnikova2016lipschitz}} studies the problem under node-DP.
In these two problems, as demonstrated in our experiments, integrating our framework with edge-DP solutions leads to an improvement in error.

For other graph analytical tasks, research has primarily focused on the edge-DP setting.  
For maximum degree estimation, only edge-DP solutions have been proposed so far. 
\reviseagain{
Under node-DP, the maximum degree or degree bound can instead be treated as a parameter and selected using techniques such as the \textit{generalized exponential mechanism}~\cite{raskhodnikova2016lipschitz}, the \textit{shifted inverse sensitivity mechanism}~\cite{fang2022shifted}, or the general framework of~\cite{liu2019private}, with the goal of optimizing the utility of a specific downstream graph query.
Note that the \textit{shifted inverse sensitivity mechanism} typically takes exponential time on this problem.
However, to the best of our knowledge, no prior work has studied the utility of maximum degree estimation itself as a standalone objective.}
Besides, densest subgraph estimation~\cite{nguyen2021differentially,farhadi2022differentially, dinitz2025almost}, and shortest path computation~\cite{sealfon2016shortest,fan2022private,chen2023differentially, deng2023differentially}, have been studied extensively under edge-DP but still lack corresponding solutions under node-DP. 
For synthetic graph generation, a recent survey~\cite{liu2024pgb} reviewed 16 works, of which 14 focus on edge-DP and only 2 address node-DP.

\revise{
Besides, we identified an important earlier related work~\cite{jain2024time}
that studies reducing node-DP to edge-DP in the dynamic setting. 
Interestingly, \citet{jain2024time} uses the same clipping mechanism~\cite{blocki2013differentially} and derives the same edge distance bound for node-neighboring graphs after clipping. 
It also shares with our paper the idea of using graph clipping to map node-neighboring graphs to graphs with bounded edge distance, thereby enabling the use of edge-DP mechanisms.
The main difference between our work and theirs lies in how the clipping threshold is chosen.
In our work, we develop an LP-based mechanism for selecting a threshold that guarantees bounded realized node-to-edge stability and attains optimal clipping error for arbitrary graph queries in the static setting. 
In contrast,~\cite{jain2024time} treats the degree bound as a user-specified parameter and develops an online Propose-Test-Release framework that privately tests whether the graph stream satisfies a safety condition under the given parameters in linear time. 
}


\section{Preliminaries}
\label{sec:prelim}

\subsection{Notation}
\label{subsec:prelim-notation}
We denote a graph as $G = (V, E)$, where $V$ is the set of nodes and $E$ is the set of edges.
In this paper, we focus on undirected graphs. 
Each edge is denoted as a pair $(u, v)$, where $u$ and $v$ are the endpoint nodes of the edge. 
Let $N=|V|$ and $M=|E|$ denote the number of nodes and edges in a graph $G$, respectively. 
$\widehat{N}$ is a predefined upper bound of $N$.
For a node $v \in V$ in graph $G$, $E(v)$ represents the set of edges incident to $v$, and $\text{deg}_G(v)$ represents the degree of $v$.
The maximum degree in $G$ is given by $\text{deg}(G) = \max_{v \in V} \text{deg}_G(v)$. 
We denote \(\deg^k(G)\) as 
the $k$-th largest degree of $G$,
and $N_\tau(G)$ as the number of nodes in $G$ with degree~$\geq \tau$.
The graph dataset $\mathcal{G}$ is a finite collection of undirected graphs.


For two graphs $G$ and $G'$, the edge distance $d_{\text{edge}}(G, G')$ is the minimum number of edge insertions or deletions required to transform $G$ into $G'$. 
We write $G \sim_{(u,v)} G'$ to denote $d_{\text{edge}}(G, G') = 1$, where $(u,v)$ is the edge that differs between $G$ and $G'$.
Similarly, we define the node distance $d_{\text{node}}(G, G')$ and write $G \sim_v G'$ if $G$ and $G'$ differ by exactly one node. 
We write $G \subseteq G'$ if $G'$ contains $G$, i.e., $G$ can be obtained by deleting a subset of nodes from $G'$ along with all edges connected to the nodes in the subset.


\subsection{Differential Privacy}
\label{subsec:prelim-dp}

\begin{definition}[Differential Privacy]
\label{def:dp}
A randomized mechanism $\mathcal{M}: \mathcal{G} \to \mathcal{R}$ satisfies $(\varepsilon, \delta)$-differential privacy if, for any two neighboring graphs $G, G' \in \mathcal{G}$, and for any subset of outputs $\mathcal{S}$ in the output space of mechanism $\mathcal{M}$, the following holds:
\[
\Pr[\mathcal{M}(G) \in \mathcal{S}] \leq e^\varepsilon \Pr[\mathcal{M}(G') \in \mathcal{S}] + \delta.
\]
$\varepsilon$ is usually a constant between $0.1$ to $10$ while $\delta\ll 1/N$.
If $\delta = 0$, the mechanism satisfies $(\varepsilon,0)$-DP or simply $\varepsilon$-DP.
\end{definition}

In the context of graph data, two graphs are considered neighboring if one can be obtained by deleting a single edge from the other, denoted as $G \sim_{\text{edge}} G'$, or by deleting a single node along with its incident edges, denoted as $G \sim_{\text{node}} G'$.
These two notions, referred to as edge-DP and node-DP~\cite{hay2009accurate}, provide privacy guarantees at different levels of granularity.
With such definitions, under node-DP, the number of nodes $N$ is private.



Differential privacy shares some desired properties:

\begin{lemma}[Post-Processing]
\label{lem:pp}
Let $\mathcal{M}: \mathcal{G} \to \mathcal{R}$ be a randomized mechanism that satisfies $(\varepsilon, \delta)$-differential privacy, and let $\mathcal{M'}: \mathcal{R} \to \mathcal{Z}$ be any deterministic or randomized mechanism. Then, $\mathcal{M'}(\mathcal{M}(G))$ satisfies $(\varepsilon, \delta)$-differential privacy for any graph $G \in \mathcal{G}$.
\end{lemma}

\begin{lemma}[Sequential Composition]
\label{lem:bc}
Let $\mathcal{M}_1, \mathcal{M}_2, \dots, \mathcal{M}_k$ be a series of randomized mechanisms, 
where for each $i \in \{1, 2, \dots, k\}$,
$\mathcal{M}_i: \mathcal{G} \to \mathcal{R}_i$ satisfies $(\varepsilon, \delta)$-differential privacy. 
The composition mechanism defined as
\[
\mathcal{M}(G) = (\mathcal{M}_1(G), \dots, \mathcal{M}_k(G))
\]
satisfies $(k\varepsilon, k\delta)$-differential privacy.
\end{lemma}

\begin{lemma}[Group Privacy]
\label{lem:gp}
If a mechanism $\mathcal{M}$ satisfies $(\varepsilon, \delta)$-differential privacy, then for any two graphs $G$ and $G'$ with $d(G, G') = \lambda$ (i.e., differing by $\lambda$ edges or nodes), and for any subset of outputs $\mathcal{S} \subseteq \mathcal{R}$, the following holds:
\[
\Pr[\mathcal{M}(D) \in \mathcal{S}] \leq e^{\lambda\varepsilon} \Pr[\mathcal{M}(D') \in \mathcal{S}] + \lambda\delta.
\]
The mechanism $\mathcal{M}$ satisfies $(\lambda\varepsilon, \lambda\delta)$-differential privacy for the group of $\lambda$ edges or nodes. 
\end{lemma}


\subsection{Common DP Mechanisms}
\label{subsec:prelim-mech}

One of the most commonly used DP mechanisms is the \textit{Laplace mechanism}. 
Here, given a query \( Q: \mathcal{G} \to \mathcal{R} \), its \textit{global sensitivity} is $\text{GS}_Q = \max_{d(G, G') = 1} \|Q(G) - Q(G')\|_1$.
For edge-DP and node-DP, we define \( d(G, G') \) as \( d_{\text{edge}} \) and \( d_{\text{node}} \), respectively.




\begin{definition}[Laplace Mechanism~\cite{dwork2006calibrating}]
\label{def:laplace}
Let \( Q: \mathcal{G} \to \mathcal{R} \) be a query function with global sensitivity \( \text{GS}_Q \). 
The mechanism 
\[
\mathcal{M}_Q(G) = Q(G) + \mathrm{Lap}(\frac{\text{GS}_Q}{\varepsilon})
\]
preserves $\varepsilon$-DP, where $\mathrm{Lap}({\text{GS}_Q}/{\varepsilon})$ is a random variable drawn from the Laplace distribution with scale of ${\text{GS}_Q}/{\varepsilon}$.
\end{definition}

Besides the Laplace mechanism, another widely used DP mechanism is the \textit{sparse vector technique} (SVT).
Given a (potentially infinite) sequence of queries \( Q_\ell: \mathcal{G} \to \mathcal{R} \) (\( \ell = 1, 2, \dots \)) with \( GS_{Q_\ell} = 1 \), SVT identifies the index of the first query whose result exceeds a threshold $\tilde{T}$.
The detailed procedure is shown in Algorithm~\ref{alg:svt}. 

\begin{algorithm}[htbp]
\caption{SVT}
\label{alg:svt}
\SetAlgoLined
\DontPrintSemicolon
\SetNoFillComment 
\KwIn{Graph $G$, threshold $T$, a sequence of queries $\{Q_\ell(G)\}$, privacy budget $\varepsilon$}
\KwOut{The first index $\ell$ where the query result exceeds $\tilde{T}$}

\tcc{Compute the noisy threshold}
$\tilde{T} \gets T + \text{Lap}(2 / \varepsilon)$\

\For{$\ell = 1, 2, \dots$}{
   $\tilde{Q}_\ell(G) \gets Q_\ell(G) + \text{Lap}(2 c / \varepsilon)$\;
    \If{$\tilde{Q}_\ell(G) > \tilde{T}$}{
        \Return $\ell$\; 
    }
}
\end{algorithm}

In the original version of SVT~\cite{dwork2009complexity}, the parameter $c=2$, meaning that the noise added to each query result is twice the scale of the noise added to the threshold. 
Very recently, \cite{liu2024unleash} points out that in specific cases, the noise added to the query result can be reduced by half. 
To formalize this, we first introduce the definition of sensitivity-monotonicity.
\begin{definition}
A function $f(G)$ is \textit{sensitivity-monotonic} if:
Given any two neighboring graphs $G\sim G'$, if $G \subseteq G'$, then $f(G) \geq f(G')$. 
\end{definition}

Then, the following lemma holds for SVT:


\begin{lemma}[\cite{liu2024unleash}]
\label{lem:svt}
If the sequence of queries $\{Q_{\ell}\}$ used in SVT are sensitivity-monotonic, given any $\varepsilon>0$, by setting $c=1$, SVT preserves $\varepsilon$-DP.
Moreover, given an error failure probability \(\beta\), if there exists some \(Q_k\) such that  
$
Q_k(G) \geq T + 4 \ln(2 / \beta)/\varepsilon,
$
then with probability at least \(1 - \beta\), SVT returns an index $\ell$ such that
\[
\ell \leq k \quad \text{and} \quad Q_\ell(G) \geq T - \frac{4}{\varepsilon}\ln\frac{2k}{\beta}.
\]
\end{lemma}

\subsection{Optimality in Error of Node-DP Mechanisms}
\label{subsec:prelim-optimality}

It is well known that the Laplace mechanism achieves \textit{worst-case optimal} error~\cite{dong2022r2t}, meaning it attains the optimal error when considering only the worst-case scenario. 
However, in graph data analytics, particularly under node-DP, worst-case optimal error is meaningless unless a strict predefined upper bound on each node's contribution to the analytical result, i.e., $\mathrm{GS}_Q$ is provided. 
Take edge counting under node-DP as an example, where $\mathrm{GS}_Q=\widehat{N}$. 
Recall that $\widehat{N}$ is a predefined upper bound for the number of nodes.
It should be set conservatively large enough to accommodate all possible input graphs.

To address this issue, \textit{down neighborhood optimality} has been introduced to quantify the error for a node-DP mechanism~\cite{dong2022r2t}. 

\begin{definition}[Down Neighborhood Optimality]
\label{def:no}
Given query $Q$,
and let \(\mathbb{M}\) denote the set of all \((\varepsilon, \delta)\)-node-DP mechanisms, 
a node-DP mechanism \( \mathcal{M}_Q \in \mathbb{M} \) is \((\rho, c)\)-down neighborhood optimal if, for any graph \( G \),
\[
\Pr\left[|\mathcal{M}_Q(G) - Q(G)| \leq c \cdot \mathcal{L}(G, \rho)\right] \geq \frac{2}{3},
\]
where \( c \) is the optimality ratio, and 
\begin{align*}
\mathcal{L}(G, \rho) := &\min_{\mathcal{M}' \in \mathbb{M}} \max_{\substack{G' \subseteq G},d_{\text{node}}(G, G') \leq \rho} 
\left\{ \xi : \Pr[|M'(G') - Q(G')| \leq \xi] \geq \frac{2}{3} \right\}.
\end{align*}

\end{definition}

Roughly speaking, for any graph $G$, we consider its $\rho$-down neighborhood, i.e., $G'\subseteq G$, $d_{\mathrm{node}}(G,G')\leq \rho$. A mechanism $\mathcal{M}_Q$ is said to be $\rho$-down neighborhood optimal if for every graph $G$, $\mathcal{M}_Q$ performs as well as the optimal mechanism $\mathcal{M}_Q'$ that is specifically designed for $G$ and its $\rho$-down neighborhood. It is clear that smaller $\rho$ and $c$ imply stronger optimality. 
When \(\rho = 0\), this reduces to \textit{instance optimality} with \(\mathcal{L}(G, 0) \equiv 0\), since for any $G$, we can always find a tailor-made mechanism  \( \mathcal{M}'(\cdot) = Q(G) \) which achieves $0$ error on \( G \).

$\mathcal{L}(G, \rho)$ can further be shown to be lower bounded by  \textit{$(\rho-1)$-downward query difference}:

\begin{lemma}[\cite{fang2022shifted}]
\label{lem:dno}
For any query $Q$, any $\rho \geq 1$, any $\varepsilon \leq \ln2$, and any graph $G$, we have
$$
\mathcal{L}(G, \rho) \geq \frac{1}{2\rho} 
\Delta Q^{(\rho-1)}(G),
$$
where 
\[\Delta Q^{(\rho-1)}(G)=\max_{\substack{G' \subseteq G,  d_{\text{node}}(G, G') \leq \rho}} \big|Q(G)-Q(G')\big|\]
is the \textit{$(\rho-1)$-downward query difference}.
\end{lemma}

Therefore, if a node-DP mechanism achieves an error of $c\cdot \Delta Q^{(\rho-1)}(G)$, then it is $(\rho,2c\rho)$-down neighborhood optimal. 
Taking edge counting under node-DP as example, $\Delta Q^{(0)}(G) = \mathrm{deg}(G)$, which is much smaller than $\widehat{N}$ for most graphs.
By convention, a mechanism is said to be optimal if it is $(\rho,2c\rho)$-down neighborhood optimal with $\rho = \mathrm{polylog}$ and $c$ is a constant.

\section{Methodology}
\label{sec:method}

In this section, we present our general framework, N2E, for reducing node-DP graph analytical tasks to edge-DP ones. 
We begin in Section~\ref{subsec:method-strawman} with a strawman solution based on group privacy.
Such an idea suffers from utility issues due to its dependency on the node count upper bound $\widehat{N}$. 
However, it offers valuable insights for the design of our mechanism.
To address these issues, we propose to use a clipping mechanism to bound individual node contributions based on a predefined clipping threshold.
In Section~\ref{subsec:method-clipping}, we introduce a distance-preserving clipping algorithm, followed by a node-DP maximum degree approximation algorithm in Section~\ref{subsec:method-deg_approx} to determine the clipping threshold. Finally, in Section~\ref{subsec:method-combine}, we combine these steps into a general framework that achieves instance-level error.

\subsection{Strawman and Roadmap}  
\label{subsec:method-strawman}

A natural approach to reducing node-DP tasks to edge-DP is to leverage group privacy. Given a graph analytical task $Q$, an $(\varepsilon, \delta)$-edge-DP mechanism $\mathcal{M}_Q^{\mathrm{edge}}(G, \varepsilon, \delta)$, an input graph $G$, and a node number upper bound $\widehat{N}$ for $G$, group privacy (Lemma~\ref{lem:gp}) guarantees that $\mathcal{M}_Q^{\mathrm{edge}}(G, \varepsilon/{\widehat{N}}, \delta/\widehat{N})$ satisfies $(\varepsilon, \delta)$-node-DP.
However, this approach breaks the utility.
Since most edge-DP mechanisms exhibit errors growing linearly with $1/\varepsilon$, dividing the privacy budget by $\widehat{N}$ amplifies the error by a factor of $\widehat{N}$. 
As shown in the introduction, for edge counting, this approach yields an error of $O({\widehat{N}}/{\varepsilon})$.
In contrast, the state-of-the-art node-DP mechanism for edge counting achieves an error of only $\tilde{O}(\deg(G)/\varepsilon)$.


A common approach to address this issue is to employ the clipping mechanism.
Here, we first clip $G$ with a degree upper bound $\tau$, i.e., all nodes after the clipping should have a degree at most $\tau$.
Besides, we require that after clipping, any two node-DP neighboring graphs should have an edge distance bounded by $\tilde{O}(\tau)$.
Then, invoking the edge-DP mechanism on the clipped graph with privacy budgets divided by the edge distance we derived, we obtain a node-DP protocol. 
When $\tau$ approaches $\deg(G)$, we eventually reduce the error dependency from $\widehat{N}$ to $\deg(G)$.

Following the roadmap, we identify two main challenges, that have not been addressed under node-DP:
\begin{itemize}
    \item[1.] Design a distance-preserving clipping mechanism: 
    Given any graph $G$ and a degree upper bound $\tau$, after the clipping, two node-DP neighboring graphs have an edge distance bounded by $g(G, \tau)$. When $\tau$ is close to $\mathrm{deg}(G)$, $g(G, \tau)$ should be well bounded by $\tilde{O}(\tau)$.
    \item[2.] Develop a node-DP mechanism to approximate the maximum degree.
\end{itemize}

\subsection{Node-DP Distance-Preserving Clipping Mechanism}
\label{subsec:method-clipping}

To address the first challenge in our roadmap, we first present a distance-preserving clipping mechanism under node-DP. 
So far, there are several distance-preserving clipping mechanisms working under edge-DP, but they can easily be shown to lose the distance-preserving property under node-DP.
The first example is a naive clipping algorithm that removes all edges from nodes with degree~$> \tau$. 
Under edge-DP, this approach ensures that the edge distance between neighboring graphs after clipping is at most \( 2\tau \). 
However, under node-DP, such an upper bound does not hold.
Suppose $G$ is a graph where all $N$ nodes have degree exactly $\tau$, and its neighboring graph $G'$ is constructed by adding a new node connected to all $N$ nodes in $G$. 
After clipping, $G$ remains unchanged, while $G'$ becomes an empty graph (since all nodes in $G'$ have degree $\tau + 1 > \tau$). 
Consequently, the edge distance between clipped $G$ and $G'$ is ${N \cdot \tau}/{2}$, which violates the distance-preserving requirement.

\begin{algorithm}[htbp]
\caption{Distance-Preserving Graph Clipping}
\label{alg:clip_graph}
\SetAlgoLined
\LinesNumbered
\DontPrintSemicolon
\SetNoFillComment 
\KwIn{Graph $G = (V, E)$, degree upper bound $\tau$}
\KwOut{Clipped graph $\mathrm{Clip}(G, \tau)$}

\tcc{Assume a natural order of edges in $E$}
$\Lambda = \langle e_1, e_2, \dots, e_n \rangle$\;

$E' \gets \emptyset$\;

\For{each node $v \in V$}{
    $E_v^\tau \gets \{ e \in E(v) \mid  \text{rank}_\Lambda(e \mid E(v)) \leq \tau \}$\;
}

\For{each edge $e = (u, v) \in E$}{
    \If{$e \in E_u^\tau$ \textbf{and} $e \in E_v^\tau$}{
        $E' \gets E' \cup \{e\}$\;
    }
}

\Return $\mathrm{Clip}(G, \tau) = (V, E')$
\end{algorithm}

Another commonly used edge-DP clipping mechanism is shown in Algorithm~\ref{alg:clip_graph}~\revise{\cite{blocki2013differentially}}. 
\blue{
Here, we first establish a natural ordering of edges. 
In practice, nodes in real-world graphs (e.g., social networks) typically have unique identifiers (e.g., user IDs), which can be used to define a natural node ordering. 
Edges are then ordered lexicographically by representing each edge with the lower-ordered node first. 
This edge ordering remains consistent between any pair of neighboring graphs and is also adopted in prior work such as~\cite{day2016publishing}.
Based on this ordering, the incident edges of each node form an ordered list, and the first $\tau$ edges refer to the top $\tau$ incident edges in that list.
We retain an edge only if it appears among the first $\tau$ edges of both its endpoints.}
Under edge-DP, this mechanism ensures that the edge distance is bounded by $3$, as each edge, aside from itself, can only affect the existence of at most one edge for each of its endpoints.
However, this distance-preserving property does not hold under node-DP.
Consider the pair of neighboring graphs shown in Figure~\ref{fig:newclip}.
Graph $G$ is a cycle over $N$ nodes $v_1$ to $v_N$, 
while its neighboring graph $G'$ adds a new node $v_0$ connected to all $N$ nodes in $G$.  
Nodes are ordered by their subscripts, i.e., $v_0 \prec v_1 \prec \dots \prec v_N$, and edges follow the lexicographic order described above.
When clipping with $\tau = 2$, the retained edges for both graphs are highlighted in blue in the figure.
Graph $G$ retains all edges, while $G'$ retains only three: $(v_0, v_1)$, $(v_0, v_2)$, and $(v_1, v_2)$. 
This results in a large edge distance of $O(N)$ between $\text{Clip}(G, 2)$ and $\text{Clip}(G', 2)$.

\begin{figure}[htbp]
  \centering
  \includegraphics[width=0.58\linewidth]{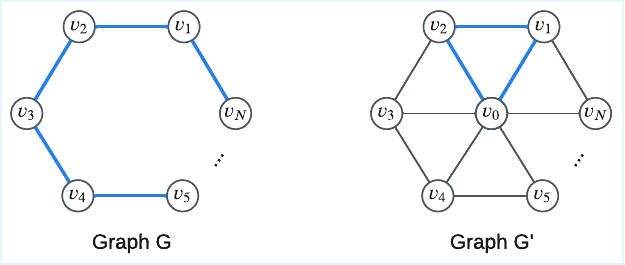}
  \caption{\blue{An example of neighboring graphs where the distance-preserving property does not hold under node-DP.}}
  \label{fig:newclip}
\end{figure}

Even though Algorithm~\ref{alg:clip_graph} is not distance-preserving for arbitrary values of $\tau$, an observation is that, given a graph $G$,
when only a small number of nodes have degree at least $\tau$,
the edge distance between $G$ and any its node-DP neighboring graph $G'$ remains well bounded after the clipping.
More precisely, if $\tau$ is the $k$-th largest degree of graph $G$, then the edge distance is bounded by $\tau + k$.
\revise{
This result also appears in~\cite{jain2024time}, where it is formalized under the name \emph{node-to-edge stability}, and is therefore not a first-time contribution of this paper.
We keep the following discussion for completeness and to motivate our main contribution, the node-DP maximum degree approximator developed later in the paper.
}
The high-level idea behind this bound is to classify all nodes in $G$ into two classes: saturated ones, with degree~$\geq \tau$, and unsaturated ones.
The number of saturated nodes is $k$.
Suppose $G$'s neighboring graph $G'$ has an additional node $u$.  
 When we do clipping, $u$ retains at most $\tau$ edges, introducing at most $\tau$ additional edges in $G'$.
Furthermore, each edge from $u$ to a saturated node can replace at most one edge in $G$, resulting in at most $k$ fewer edges in $G'$.

For example, let us consider two node-DP neighboring graphs $G$ and $G'$ as shown in Figure~\ref{fig:clip}. 
Graph $G$ consists of $1000$ $3$-stars and $3$ $4$-stars, and $G'$ differs from $G$ by adding a node $u$ that connects to the center of every star in $G$.
The saturated nodes of $G$ are highlighted in red in the figure.
Applying Algorithm~\ref{alg:clip_graph} to clip both graphs with $\tau = 4$, 
we have $k = N_{\tau}(G) = 3$, since only $3$ nodes in $G$ are saturated with $\geq 4$ edges.
After clipping, the node $u$ can retain up to $4$ edges not present in $G$. 
Moreover, each $4$-star center differs by at most $1$ edge from its original $4$ edges between $G$ and $G'$, introducing an edge distance up to $k = 3$.
As a result, the edge distance between $G$ and $G'$ after clipping is bounded by $\tau + k = 7$.

\begin{figure}[htbp]
    \centering
    \includegraphics[width=0.65\linewidth]{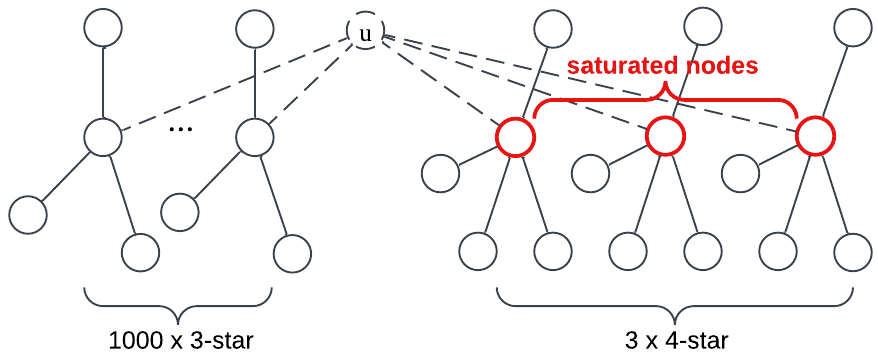} 
    \caption{Demonstration of the distance-preserving property for our node-DP clipping mechanism.}
    \label{fig:clip}
\end{figure}


The formal statement is as below.

\begin{lemma}
\label{lem:clip_edge_dis}
Given any two neighboring graphs $G$ and $G'$ such that $G \sim_{\text{\emph{node}}} G'$, 
let $\tau = \deg^k(G)$.
The edge distance between $\text{\emph{Clip}}(G, \tau)$ and $\text{\emph{Clip}}(G', \tau)$ is bounded by $\tau + k$, where $\text{\emph{Clip}}(G, \tau)$ and $\text{\emph{Clip}}(G', \tau)$ are clipped graphs with Algorithm~\ref{alg:clip_graph}.
\end{lemma}



\begin{proof}
Without loss of generality, we assume $G \subseteq G'$ and let the differing node be $u$. 
Let $G = (V, E)$ and $G' = (V', E')$.
We use $\overline{G} = (\overline{V}, \overline{E})$ to denote $\text{Clip}(G, \tau)$,
and use $\overline{G'}= (\overline{V'}, \overline{E'})$
to denote $\text{Clip}(G', \tau)$. 
By definition of $\deg^k(G)$, $G$ has at most $k$ nodes with degree $\geq \tau$.

Consider all different edges between $\overline{G}$ and $\overline{G'}$:
\begin{itemize}
    \item[1.] \textbf{The number of edges $e$ that $e \in \overline{E'}$ but $e \notin \overline{E}$ is at most $\tau$.}

    If there exists an edge $e = (v_1, v_2) \in \overline{E'}$ and $v_1, v_2 \in V$, then $e$ is among the $\tau$ smallest-order edges for $v_1$ and $v_2$ in both $G'$ and $G$, which indicates that $e \in \overline{E}$.
    Therefore, such edge $e = (v_1, v_2)$ in this category has $v_1 = u$ or $v_2= u$. Since $u$ can keep at most $\tau$ edges after clipping, the number of this kind of edges is bounded by $\tau$. 
   
    \item[2.] \textbf{The number of edges $e$ that $e \in \overline{E}$ but $e \notin \overline{E'}$ is at most $k$.}

    For any such edge $e = (v_1, v_2) \in \overline{E}$, $e$ is among the $\tau$ smallest-order for both $v_1$ and $v_2$ in $G$. 
    The condition $e \notin \overline{E'}$ occurs only if $v_1$ or $v_2$ is a saturated node in $G$, and a new edge $e' = (u, v) \in E'\text{ where } v \in \{v_1, v_2\}$ replaces $e$ in $\overline{G'}$ as one of the $\tau$ smallest-order edges for that node.
    Since the number of saturated nodes in $G$ is bounded by $k$, the number of edges in this category is at most $k$. 
\end{itemize}

To conclude, the total edge difference between $\overline{G}$ and $\overline{G'}$ is bounded by $\tau + k$.
\end{proof}

Therefore, using $\tau$ as the degree upper bound in Algorithm~\ref{alg:clip_graph} to do the clipping can bound the edge distance between $G$ and its neighboring graph by
\begin{equation}
\label{e:n_contribution}
g(G, \tau) = \tau + N_\tau(G)
\end{equation}
where
$N_\tau(G)$ is how many nodes in $G$ with degree at least $\tau$.  
Then, we can further invoke any edge-DP mechanism on the clipped graph with privacy budgets divided by $g(G, \tau)$ to achieve node-DP.

However, this approach has two problems. 
First, the value of $g(G, \tau)$ depends on how well $\tau$ approximates $\mathrm{deg}(G)$. 
A $\tau$ that is too small or too large yields a large $g(G, \tau)$.
Second, $g(G, \tau)$ is instance-specific and can differ significantly between $G$ and its neighboring graph $G'$. 
For example, consider an empty graph $G$, and we insert one node connected to all nodes to obtain $G'$. Then, $g(G,1)=1$ while $g(G',1)=N+2$.
Therefore, directly using $g(G, \tau)$ to allocate the privacy budget would break the privacy~\cite{nissim2007smooth}.
This provides key insights for designing our maximum degree approximation mechanism. 
Specifically, the output $\tau^*$ should satisfy the following three properties:
\begin{itemize}
    \item[1.] $\tau^*$ is not too large, i.e., $\tau^* = \tilde{O}(\mathrm{deg}(G))$.
    \item[2.] $\tau^*$ is not too small so that $N_{\tau^*}(G)$ can be well bounded.
    \item[3.] $\tau^*$ derives a DP-compliant upper bound of $g(G, \tau^*)$ that can be used in group privacy. 
\end{itemize}


Next, we explore approaches to determine $\tau^*$ that satisfies all three properties while preserving node-DP.

\subsection{Maximum Degree Estimation}
\label{subsec:method-deg_approx}

Before discussing the node-DP solution to obtain a good $\tau^*$ that satisfies the three listed properties, let us start with developing a corresponding edge-DP solution. 

\subsubsection{Edge-DP Maximum Degree Estimation}
\label{subsubsec:method-deg-edgedp}
The most standard approach to realize edge-DP maximum degree estimation is to use the Laplace mechanism. Since any edge can influence the graph’s maximum degree by up to $1$, one can obtain an edge-DP maximum degree estimation by adding a noise of scale $1/\varepsilon$ to $\mathrm{deg}(G)$. 
However, the result estimation $\tau^*$ fails to satisfy our second target property. 
For example, assume graph $G$ has all $N$ nodes have the same degree. 
Using the Laplace mechanism, the output $\tau^*$ satisfies $\tau^* \leq \mathrm{deg}(G)$ under a probability of $1/2$. 
Therefore, under the same probability, $N_{\tau^*}(G) = N$, as all $N$ nodes have their degree higher than $\tau^*$. 

This motivates the design of a new edge-DP maximum degree estimation mechanism that outputs $\tau^*$ with the above $3$ properties.
Instead of using additive error to ensure privacy, we resort to another privacy-preserving algorithm, SVT, to produce the expected $\tau^*$.
We propose the following query function used in the SVT:

\begin{equation}
\label{e:clip_degree_num}
    Q_{\text{Del-Deg}}(G, \tau) = -\frac{1}{2}  
    \sum_{\substack{v \in V \\ \deg_G(v) \geq \tau}} (\deg_G(v) - \tau).
\end{equation}
Generally speaking, we sum the node degrees exceeding $\tau$ in $G$, 
and $Q_{\text{Del-Deg}}(G, \tau)$ denotes the opposite of half the degree sum.
$Q_{\text{Del-Deg}}(G, \tau)$ has three properties as below. 

\begin{lemma} 
\label{lem:clip_deg_num}
$Q_{\text{\emph{Del-Deg}}}(G, \tau)$ has the following properties:
\begin{itemize}
    \item[1.] \textbf{Monotonicity}: $Q_{\text{\emph{Del-Deg}}}(G, \tau)$ is monotonic and becomes $0$ when $\tau \geq \mathrm{deg}(G)$.
    \item[2.] \textbf{Bounded Sensitivity}: $Q_{\text{\emph{Del-Deg}}}(G, \tau)$ has $\text{\emph{GS}}_Q = 1$ under edge-DP.
    \item[3.] \textbf{Sensitivity-Monotonicity}: 
    For any two neighboring graphs $G \sim_{\text{\emph{edge}}} G'$ and any $\tau$,\newline
    $Q_{\text{\emph{Del-Deg}}}(G, \tau) \geq Q_{\text{\emph{Del-Deg}}}(G', \tau)$ if $G \subseteq G'$.
\end{itemize} 
\end{lemma}

\begin{proof}
The proofs for all properties are provided below.

\begin{itemize}
    \item[1.] \textbf{Monotonicity}: The sum of node degrees exceeding $\tau$ decreases as $\tau$ increases. When $\tau \geq \mathrm{deg}(G)$, no node's degree exceeds $\tau$. Since $Q_{\text{Del-Deg}}(G, \tau)$ is defined as $-1/2$ times this sum, it increases with $\tau$ and becomes $0$ when $\tau \geq \mathrm{deg}(G)$.

    \item[2.] \textbf{Bounded Sensitivity}: Adding or removing an edge influences the degree of two nodes by $1$, thus altering the exceeding degree by at most $2$. Therefore, the global sensitivity of $Q_{\text{Del-Deg}}(G, \tau)$ is $1$.

    \item[3.] \textbf{Sensitivity Monotonicity}: If $G'$ differs from $G$ by one additional edge, then $\deg_{G'}(v) \geq \deg_G(v)$ for all $v \in V$. 
    Consequently, $Q_{\text{Del-Deg}}(G', \tau) \leq Q_{\text{Del-Deg}}(G, \tau)$. \qedhere
\end{itemize}
\end{proof}

Using $Q_{\text{Del-Deg}}(G, \tau)$ as the query function in SVT, we develop a new edge-DP maximum degree estimation mechanism as in Algorithm~\ref{alg:edge_dp_max_deg_esti}.
\orange{
The algorithm iteratively increases $\tau$ until the noisy sum of node degrees exceeding $\tau$ falls below a noisy threshold.}

\begin{algorithm}[htbp]
\caption{Edge-DP Maximum Degree Estimation}
\label{alg:edge_dp_max_deg_esti}
\SetKw{Break}{break}
\SetAlgoLined
\LinesNumbered
\DontPrintSemicolon
\KwIn{Graph $G$, privacy budget $\varepsilon$, parameter $\beta$}
\KwOut{$\tau$}

$T \gets -4 \ln(2 / \beta)/\varepsilon$   

$\tilde{T} \gets T + \text{Lap}(2 / \varepsilon)$ 

\For{$\tau = 1, 2, 3, 4,\dots$}{

    $\tilde{Q}_{\tau}(G) \gets Q_{\text{Del-Deg}}(G, \tau) + \text{Lap}(2 / \varepsilon)$  

    \If{$\tilde{Q}_{\tau}(G) > \tilde{T}$}{
        \Break\;
    }
}
 \Return $\tau$
\end{algorithm}

With Lemma~\ref{lem:clip_deg_num} and Lemma~\ref{lem:svt}, it can be shown that Algorithm~\ref{alg:edge_dp_max_deg_esti} satisfies $\varepsilon$-edge-DP. 
Furthermore, Algorithm~\ref{alg:edge_dp_max_deg_esti} runs in polynomial time, as each $Q_{\text{Del-Deg}}(G, \tau)$ can be solved in $O(N)$ time.  
For the utility, we have:
\begin{lemma}
\label{lem:del_deg_bound}
Given any $\varepsilon$ and $\beta$, for any graph $G$, Algorithm~\ref{alg:edge_dp_max_deg_esti} outputs a $\tau$ such that with probability at least $1 - \beta$, $\tau \leq \mathrm{deg}(G)$, and 
\begin{equation}
 \label{e:del_deg_bound}
    |Q_{\text{\emph{Del-Deg}}}(G, \tau)| \leq \frac{4}{\varepsilon}\ln\mathrm{deg}(G) + \frac{8}{\varepsilon}\ln\frac{2}{\beta}.
\end{equation}
\end{lemma}

\begin{proof}
For $\tau \geq \mathrm{deg}(G)$, we have $Q_{\text{Del-Deg}}(G, \tau) = 0 = T + {4} \ln (2/\beta)/ {\varepsilon}$. By Lemma~\ref{lem:svt}, with probability at least $1 - \beta$, Algorithm~\ref{alg:edge_dp_max_deg_esti} returns a $\tau$ such that $\tau \leq \mathrm{deg}(G)$ and
\begin{equation*}
    Q_{\text{Del-Deg}}(G, \tau) \geq T - \frac{4}{\varepsilon} \ln \frac{2\mathrm{deg}(G)}{\beta}  = -\frac{4}{\varepsilon} \ln \mathrm{deg}(G) - \frac{8}{\varepsilon} \ln \frac{2}{\beta}.
\end{equation*}
Thus, the lemma follows. \qedhere
\end{proof}

(\ref{e:del_deg_bound}) further leads to an error of at most $\Delta Q^{(\rho-1)}(G)$ with $\rho = \tilde{O}(1/\varepsilon)$, which implies Algorithm~\ref{alg:edge_dp_max_deg_esti} is ($\tilde{O}(1/\varepsilon),\tilde{O}(1/\varepsilon)$)-down neighborhood optimal.
It can be further shown that setting $\tau^* = \tau + 1$ satisfies our three desired properties. Since our main focus is on node-DP, we do not elaborate on the details here.

\paragraph{Comparison with Laplace mechanism.}
In the worst case, the output of Algorithm~\ref{alg:edge_dp_max_deg_esti} has an additive error of $O(\log(\mathrm{deg}(G)/\beta)/\varepsilon)$.
In comparison, the Laplace mechanism achieves a tighter additive error of $O(\log(1/\beta)/\varepsilon)$, which generally yields better performance on typical networks with power-law degree distributions.
However, the key advantage of Algorithm~\ref{alg:edge_dp_max_deg_esti} lies in achieving low rank error rather than low additive error.
As we will discuss later, when integrated into our framework to upgrade any edge-DP solution to provide node-DP, the error is amplified by a factor of $\tilde{O}(\deg(G))$.
As a result, the Laplace mechanism incurs an additive error of $\tilde{O}(\deg(G)/\varepsilon)$, which breaks the utility.
In contrast, Algorithm~\ref{alg:edge_dp_max_deg_esti} achieves a comparable $\tilde{O}(\deg(G)/\varepsilon)$ rank error and still preserves strong utility guarantees.

\subsubsection{Node-DP Exponential-Time Maximum Degree Approximation}
\label{subsubsec:method-deg-nodedp_exp}

So far, we have an effective solution for maximum degree estimation that satisfies three desired properties under edge-DP. However, directly extending this approach to node-DP is problematic.
The main issue is that the property of bounded sensitivity for $Q_{\text{Del-Deg}}(G, \tau)$ does not hold under node-DP. 
Since two neighboring graphs differing by a node can have a degree sum difference of up to $2N$, the global sensitivity for this query becomes $N$ under node-DP.
Therefore, a query function satisfying the three desired properties under node-DP is required.

Here, we introduce:
\begin{equation}
 \label{e:clip_node_num}
    Q_{\text{Del-N}}(G, \tau) = -\min \left\{ |N_G - N_{G^*}| : G^* \subseteq G, \deg(G^*) \leq \tau \right\}, 
\end{equation}
which is the opposite of the minimum number of nodes to remove from $G$ to bound its degree with $\tau$.
We observe that $Q_{\text{Del-N}}(G, \tau)$ has three properties as below under node-DP, which is similar to $Q_{\text{Del-Deg}}(G, \tau)$ under edge-DP.

\begin{lemma}
\label{lem:clip_node_num}
   $Q_{\text{\emph{Del-N}}}(G, \tau)$ satisfies the following $3$ properties: 

\begin{itemize}
    \item[1.] \textbf{Monotonicity}: 
    $Q_{\text{\emph{Del-N}}}(G, \tau)$ is monotonic and becomes $0$ when $\tau \geq \mathrm{deg}(G)$.
    \item[2.] \textbf{Bounded Sensitivity}: 
    $Q_{\text{\emph{Del-N}}}(G, \tau)$ has $\text{\emph{GS}}_Q = 1$ under node-DP.
    \item[3.] \textbf{Sensitivity-Monotonicity}: 
    For any two neighboring graphs $G \sim_{\text{\emph{node}}} G'$ and any $\tau$,\newline
    $Q_{\text{\emph{Del-N}}}(G, \tau) \geq Q_{\text{\emph{Del-N}}}(G', \tau)$ if $G \subseteq G'$.
\end{itemize}
\end{lemma}

\begin{proof}
The proofs for all properties are provided below.

\begin{itemize}
    \item[1.] \textbf{Monotonicity}: As $\tau$ increases, fewer nodes need to be removed from $G$ to ensure all node degrees are at most $\tau$. When $\tau \geq \deg(G)$, no removals are necessary, hence $Q_{\text{Del-N}}(G, \tau) = 0$.
    Since $Q_{\text{Del-N}}(G, \tau)$ is defined as the opposite of the number of nodes to remove, its value increases with $\tau$.

    \item[2.] \textbf{Bounded Sensitivity}: Let neighboring graphs $G \sim_{v} G'$ and $G \subseteq G'$. 
    Let $S_G$ be a minimal node set to remove that bounds the degree of $G$ by $\tau$.
    Then, $S_G \cup \{v\}$ is a valid node removal set for $G'$, where $v$ is the differing node.
    Thus, $Q_{\text{Del-N}}(G', \tau) \geq Q_{\text{Del-N}}(G, \tau) + 1$. 
    Conversely, let $S_{G'}$ be a minimal node removal set for $G'$. Since $G \subseteq G'$, removing $S_{G'}$ from $G$ also suffices to bound its degree by $\tau$.
    Therefore, $Q_{\text{Del-N}}(G', \tau) \leq Q_{\text{Del-N}}(G, \tau)$.
    To conclude, the global sensitivity of $Q_{\text{Del-N}}(G, \tau)$ is bounded by $1$.

    \item[3.] \textbf{Sensitivity Monotonicity}: Follows directly from the argument in Property 2. \qedhere
\end{itemize}
\end{proof}

Following the same steps of Algorithm~\ref{alg:edge_dp_max_deg_esti}, we can incorporate $Q_{\text{Del-N}}(G, \tau)$ with SVT to obtain a maximum degree estimation $\tau$ under node-DP. 
Similar to Lemma~\ref{lem:del_deg_bound}, it follows that $\tau \leq \mathrm{deg}(G)$ and $|Q_{\text{Del-N}}(G, \tau)| = \tilde{O}(1/\varepsilon)$.
\orange{This implies that $\tau$ upper-bounds the maximum degree of $G$ after removing a minimal subset of $\tilde{O}(1/\varepsilon)$ nodes.}
However, such a $\tau$ can not directly serve as $\tau^*$, as it fails to satisfy the second required property, which requires a bounded $N_{\tau^*}(G)$. 
\orange{Specifically, $|Q_{\text{Del-N}}(G, \tau)| \leq N_{\tau}(G)$, as removing all nodes with degree at least $\tau$ yields a feasible solution to $Q_{\text{Del-N}}(G, \tau)$.  
Therefore, a bounded value of $|Q_{\text{Del-N}}(G, \tau)|$ does not necessarily imply that $N_{\tau}(G)$ is bounded.}

To construct such a $\tau^*$ satisfying our three desired properties under node-DP, we set
\begin{equation}
\label{e:exp_tau*}
    \tau^* =  \tau + |Q_{\text{Del-N}}(G, \tau)|+ \mathrm{Lap}(\frac{c}{\varepsilon}) 
    +\mathrm{polylog},
\end{equation} 
where $c$ is some constant.
The high-level idea behind this construction is as follows. 
To satisfy the second desired property, we first increment $\tau$ by $|Q_{\text{Del-N}}(G, \tau)| + 1$ to get $\tau'$.
As a result, 
\begin{equation}
\label{e:ntau_bound}
    N_{\tau'}(G) \leq |Q_{\text{Del-N}}(G, \tau)|.
\end{equation}
This follows from the definition of $Q_{\text{Del-N}}(G, \tau)$, which ensures that removing $|Q_{\text{Del-N}}(G, \tau)|$ nodes from $G$ can result in a graph $G^*$ where all node degrees are at most $\tau$.
Therefore, every node in $G^*$ has a degree at most $\tau + |Q_{\text{Del-N}}(G, \tau)|$ in $G$.
Consequently, at most $|Q_{\text{Del-N}}(G, \tau)|$ number of nodes in $G$ have a degree exceeding this threshold, ensuring that
$N_{\tau'}(G) \leq |Q_{\text{Del-N}}(G, \tau)|$. 

However, using $\tau'$ as $\tau^*$ introduces an additional privacy concern: $\tau'$ depends on $G$, as it includes the instance-specific term $|Q_{\text{Del-N}}(G, \tau)|$.  
Fortunately, $|Q_{\text{Del-N}}(G, \tau)|$ has a global sensitivity of $1$ by Lemma~\ref{lem:clip_node_num}.  
Thus, adding a Laplace noise with scale $2/\varepsilon$ to $\tau'$ ensures $\varepsilon/2$-node-DP, and we denote the result as $\tau''$.  
Since $\tau''$ is smaller than $\tau'$ with probability $1/2$, it still violates the second property. 
To address this, we further increment $\tau''$ by a logarithmic term and denote the result as $\tau^*$.  
This added logarithmic term ensures that $\tau^*$ exceeds $\tau'$ with high probability, thereby maintaining $N_{\tau^*}(G)$ bounded by $|Q_{\text{Del-N}}(G, \tau)|$. 

So far, we have already shown that $\tau^*$ satisfies the first two desired properties. For the third property, we can use $2\tau^*$ as an upper bound of $g(G,\tau^*)$ 
\orange{since $\tau^*$ already incorporates a high-probability bound on $N_{\tau^*}(G)$ through the term $|Q_{\text{Del-N}}(G, \tau)|$.}

\begin{algorithm}[htbp]
\caption{Node-DP Maximum Degree Approximation}
\label{alg:exp_node_max_deg_approx}
\SetKw{Break}{break}
\SetAlgoLined
\LinesNumbered
\DontPrintSemicolon
\KwIn{Graph $G$, privacy budgets $\varepsilon$ and $\delta$, parameter $\beta$}
\KwOut{$\tau^*$}

$T \gets -8 \ln(4 / \beta)/\varepsilon$   

$\tilde{T} \gets T + \text{Lap}(4 / \varepsilon)$ 

\For{$\tau = 1, 2, 4, 8, \dots$}{
   $\tilde{Q}_\tau(G) \gets Q_{\text{Del-N}}(G, \tau) + \text{Lap}(4 / \varepsilon)$ 
  
    \If{$\tilde{Q}_\tau(G) > \tilde{T}$}{
        \Break\;
    }
}
$\tau^*\gets\tau + |Q_{\text{Del-N}}(G, \tau)|+ \mathrm{Lap}(\frac{2}{\varepsilon}) + \frac{2}{\varepsilon}\ln\max(\frac{1}{\delta},\frac{2}{\beta}) +1$\nllabel{line:tau-postprocess}\;
\Return $\tau^*$
\end{algorithm}

Finally, our node-DP mechanism for maximum degree approximation works as follows. We first invoke a similar algorithm as Algorithm~\ref{alg:edge_dp_max_deg_esti}, but replaces $Q_{\text{Del-Deg}}(G, \tau)$ with
$Q_{\text{Del-N}}(G, \tau)$. 
Besides, we evaluate $\tau$ on a logarithmic scale of $\tau=1,2,4,8\dots$ to reduce the number of iterations, since the computation of $Q_{\text{Del-N}}(G, \tau)$ is costly.
This approach is acceptable because, according to the first property of $\tau^*$, we only require an approximation rather than an accurate estimation of the maximum degree. 
After obtaining $\tau$ from the SVT, we construct $\tau^*$ following (\ref{e:exp_tau*}) with carefully designed parameters. The detailed algorithm is shown as Algorithm \ref{alg:exp_node_max_deg_approx}.

For privacy, \red{Algorithm~\ref{alg:exp_node_max_deg_approx} satisfies $\varepsilon$-node-DP. 
Specifically,
$\tau$ obtained in the SVT preserves $\varepsilon/2$-node-DP, while the remaining $\varepsilon/2$ privacy budget is allocated to privatize $\tau^*$. 
The parameter $\delta$ in Line~\ref{line:tau-postprocess} is used for utility bounds and does not affect the privacy guarantees.}
\orange{For utility, we allocate an error failure probability of $\beta/2$ to the computation of $\tau$ and $\beta/2$ to that of $\tau^*$. 
Line~\ref{line:tau-postprocess} further ensures that $|Q_{\text{Del-N}}(G, \tau)|$ upper-bounds $N_{\tau^*}(G)$ with failure probability at most $\delta$.
The following utility guarantees hold:
}

\begin{lemma}
\label{lem:exp_clip_node_bound}
Given any $\varepsilon$, $\beta$, and $\delta$, for any graph $G$, Algorithm~\ref{alg:exp_node_max_deg_approx} returns a $\tau^*$ such that with probability at least $1-\beta$,
\begin{equation*}
    \tau^* \leq 2\mathrm{deg}(G) + \frac{8}{\varepsilon}\ln\log(4\mathrm{deg}(G)) + \frac{16}{\varepsilon}\ln\frac{4}{\beta} + \frac{4}{\varepsilon}\ln\max(\frac{1}{\delta}, \frac{2}{\beta})+ 1,
\end{equation*}
and 
\begin{equation*}
   N_{\tau^*}(G) \leq \frac{8}{\varepsilon}\ln\log(4\mathrm{deg}(G)) + \frac{16}{\varepsilon}\ln\frac{4}{\beta}.
\end{equation*}
Furthermore, with probability at least $1-\delta$, 
\begin{equation*}
    \tau^*+ N_{\tau^*}(G) \leq 2\tau^*. 
\end{equation*}
\end{lemma}

\begin{proof}
Let $\tau$ denote the output from SVT in Algorithm~\ref{alg:exp_node_max_deg_approx}.
When $\tau \geq \mathrm{deg}(G)$, $Q_{\text{Del-N}}(G, \tau) = 0 = T + {8} \ln (4/\beta)/ {\varepsilon}$. 
Since $\tau$ is evaluated on a logarithmic scale, the number of iterations required to reach $\tau \geq \mathrm{deg}(G)$ is bounded by $\ln \log (4\deg(G))$. 
Thus, by Lemma~\ref{lem:svt}, with probability at least $1 - \beta/2$, we have $\tau \leq 2\deg(G)$ and
\begin{equation*}
    |Q_{\text{Del-N}}(G, \tau)| \leq \frac{8}{\varepsilon} \ln \log(4\deg(G)) + \frac{16}{\varepsilon} \ln \frac{4}{\beta}.
\end{equation*}

Next, $\tau$ is post-processed to $\tau^*$ in Line~\ref{line:tau-postprocess} of Algorithm~\ref{alg:exp_node_max_deg_approx}.
By the tail bound of the Laplace distribution, with probability at least $1 - \min(\beta/2, \delta)$,
\begin{equation*}
    0 \leq \mathrm{Lap}(\frac{2}{\varepsilon}) + \frac{2}{\varepsilon}\ln\max(\frac{1}{\delta}, \frac{2}{\beta}) \leq \frac{4}{\varepsilon}\ln\max(\frac{1}{\delta}, \frac{2}{\beta}).
\end{equation*}
Thus, under the same probability,
\begin{equation*} 
   \tau^* \leq \tau + |Q_{\text{Del-N}}(G, \tau)| + 1 + \frac{4}{\varepsilon}\ln\max(\frac{1}{\delta}, \frac{2}{\beta}),
\end{equation*}
and
\begin{equation*}
    \tau^* \geq \tau + |Q_{\text{Del-N}}(G, \tau)| + 1.
\end{equation*}
Moreover, by (\ref{e:ntau_bound}), it follows from the above that
\begin{equation*}
    N_{\tau^*}(G) \leq |Q_{\text{Del-N}}(G, \tau)|.
\end{equation*}

By applying the union bound and incorporating the upper bounds on $\tau$ and $|Q_{\text{Del-N}}(G, \tau)|$, we conclude that, with probability at least $1 - \beta$, the bounds on $\tau^*$ and $N_{\tau^*}(G)$ hold as claimed.

Finally, with probability at least $1 - \delta$,
$
    N_{\tau^*}(G) \leq |Q_{\text{Del-N}}(G, \tau)|
$
holds. Since $\tau^*$ contains the term $|Q_{\text{Del-N}}(G, \tau)|$, it follows directly that
\begin{equation*}
    \tau^* + N_{\tau^*}(G) \leq 2\tau^*. \qedhere
\end{equation*}
\end{proof}

So far, we have obtained a node-DP maximum degree approximation $\tau^*$ satisfying all three desired properties.
However, the computation of $\tau^*$ incurs significant computational costs. 
 Specifically, computing $Q_{\text{Del-N}}(G, \tau)$ requires examining all possible combinations of node deletions, resulting in exponential runtime in the worst case.

\subsubsection{Node-DP Polynomial-Time Maximum Degree Approximation}
\label{subsubsec:method-deg-nodedp_poly}

In this section, we show how to reduce the runtime of our node-DP maximum degree approximation mechanism. 
Recall that the bottleneck lies in computing $Q_{\text{Del-N}}(G, \tau)$, which requires exponential time.
Our key idea is to approximate its value by leveraging Linear Programming (LP) 
and integrate the result with our proposed node-DP algorithm.  
We will show that the enhanced algorithm still outputs a $\tau^*$ that complies with the three desired properties.



We begin by approximating $Q_{\text{Del-N}}(G, \tau)$ with an LP formulation, denoted as $Q_{\text{LP-Del-N}}(G, \tau)$. 
In this formulation, we allow the fractional existence of nodes and edges in $G$. 
Specifically,
for each node $v \in V$, we assign a value $x_v \in [0, 1]$ to indicate whether to delete the node $v$ from $G$, i.e., $0$ represents complete retention and $1$ represents complete removal.
In contrast, for each edge $e \in E$, we assign a value $y_e$ from $[0, 1]$ to indicate whether to maintain $e$, i.e., $0$ represents removal and $1$ represents retention.
The LP $Q_{\text{LP-Del-N}}(G, \tau)$ maximizes the opposite of node removal count while ensuring constraints on node degrees and node-edge relationships, indicated by (\ref{eq:cons1}) and (\ref{eq:cons2}) respectively.
The LP formulation of $Q_{\text{LP-Del-N}}(G, \tau)$ is given by: 
\begin{align}
\text{maximize} \quad & Q_{\text{LP-Del-N}}(G, \tau) = -\sum_{v \in V} x_v \nonumber \\
\text{s.t.} \quad & y_e \geq 1 - x_{v'} - x_{v''}, \quad e = (v', v'') \in G, \label{eq:cons1} \\
                  & \sum_{e' \in E(v)} y_{e'} \leq \tau, \quad v \in V, \label{eq:cons2} \\
                   y_e &\in [0, 1], \quad e \in E, \nonumber \\
                   x_v &\in [0, 1], \quad v \in V. \nonumber
\end{align}

It is trivial to see $Q_{\text{LP-Del-N}}(G, \tau)$ is a relaxed version of $Q_{\text{Del-N}}(G, \tau)$: $Q_{\text{Del-N}}(G, \tau)$ can be formulated as an integer linear programming where each variable is from a discrete domain, i.e., $x_v \in \{1, 0\}$ and $y_e \in \{1, 0\}$, where the values have the same interpretations as in $Q_{\text{LP-Del-N}}(G, \tau)$.
Different from $Q_{\text{Del-N}}(G, \tau)$ that requires exponential runtime, $Q_{\text{LP-Del-N}}(G, \tau)$ can be solved in polynomial time.
Moreover, it retains the same properties as $Q_{\text{Del-N}}(G, \tau)$ under node-DP, as outlined below.

\begin{lemma}
\label{lem:lp_property}
   $Q_{\text{\emph{LP-Del-N}}}(G, \tau)$ satisfies the following $3$ properties under node-DP: 

\begin{itemize}
    \item[1.] \textbf{Monotonicity}: 
    $Q_{\text{\emph{LP-Del-N}}}(G, \tau)$ is monotonic and becomes $0$ when $\tau \geq \mathrm{deg}(G)$.
    \item[2.] \textbf{Bounded Sensitivity}: 
    $Q_{\text{\emph{LP-Del-N}}}(G, \tau)$ has $\text{\emph{GS}}_Q = 1$.
    \item[3.] \textbf{Sensitivity-Monotonicity}: $Q_{\text{\emph{LP-Del-N}}}(G, \tau)$ is sensitivity-monotonic.
\end{itemize}
\end{lemma}

\begin{proof}
Assume that $\{x_v\}$ and $\{y_e\}$ is the optimal solution for $Q_{\text{LP-Del-N}}(G, \tau)$.
 The proofs of all properties are given below.
 
\begin{itemize}
     \item[1.] \textbf{Monotonicity}: As $\tau$ increases, the values of $y_e$ can increase, leading to lower values of $x_v$. 
     When $\tau \geq \mathrm{deg}(G)$, all edges can be retained, resulting in $x_v =0$ for all $v \in V$.
    Since $Q_{\text{LP-Del-N}}(G, \tau)$ is the opposite of the sum of $x_v$, it increases with $\tau$ and becomes $0$ when $\tau \geq \mathrm{deg}(G)$.

     \item[2.] \textbf{Bounded Sensitivity}:
    Let neighboring graphs $G \sim _u G'$ such that $G \subseteq G'$ and $G' = (V', E')$. 
    A feasible solution for $G'$ can be constructed by keeping the values of $\{x_v\}$ and $\{y_e\}$ from $G$, and set $x_u =1$ and $y_e' = 0$ for $e \in E'(u)$. 
    Thus, $Q_{\text{LP-Del-N}}(G, \tau)-1 \leq Q_{\text{LP-Del-N}}(G', \tau)$.
    Conversely, a feasible solution for $G$ can be constructed by applying the assignments from the optimal solution of $G'$ to the corresponding nodes and edges in $G$. Therefore, $Q_{\text{LP-Del-N}}(G', \tau) \leq Q_{\text{LP-Del-N}}(G, \tau)$.
    To conclude, $Q_{\text{LP-Del-N}}(G, \tau)$ has a global sensitivity of $1$.

     \item[3.] \textbf{Sensitivity-Monotonicity}:   
     Follows directly from the argument in Property 2.  \qedhere
\end{itemize}
\end{proof}

These three properties ensure that by replacing $Q_{\text{Del-N}}(G, \tau)$ with $Q_{\text{LP-Del-N}}(G, \tau)$, Algorithm~\ref{alg:exp_node_max_deg_approx} still preserves $\varepsilon$-node-DP. 
However, the utility guarantee for the output $\tau^*$ no longer holds. 
Since $Q_{\text{LP-Del-N}}(G, \tau)$ is a relaxation of $Q_{\text{Del-N}}(G, \tau)$, its value is greater than $Q_{\text{Del-N}}(G, \tau)$. 
As a result, when using $Q_{\text{LP-Del-N}}(G, \tau)$, the SVT may stop earlier with a smaller $\tau$. 
While such an $\tau$ ensures a bounded $|Q_{\text{LP-Del-N}}(G, \tau)|$, it does not necessarily imply a bounded $|Q_{\text{Del-N}}(G, \tau)|$, which is essential to satisfy the second property.  

To address this issue, we establish a relationship between \newline
$Q_{\text{LP-Del-N}}(G, \tau)$ and $Q_{\text{Del-N}}(G, \tau)$. The high-level idea is as follows.
Let $\{x_v\}$ and $\{y_e\}$ denote the optimal solution to the LP $Q_{\text{LP-Del-N}}(G, \tau)$.
Recall that $x_v$ represents the fractional removal of node $v$ from $G$, and $y_e$ represents the fractional retention of edge $e$.
Now we construct a graph $G'$ by removing all nodes $v$ with $x_v > 1/3$ from $G$. 
On one hand, since only nodes with $x_v > 1/3$ are removed, the total number of deleted nodes is at most $3$ times $\sum_{v \in V} x_v$, which is equivalent to $3\cdot|Q_{\text{LP-Del-N}}(G, \tau)|$. 
On the other hand, for every edge in $G'$, both endpoints $v$ and $v'$ satisfy $x_v,x_{v'}\leq1/3$, which implies $y_e\geq1/3$. Therefore, it can be inferred that the maximum degree of $G'$ is bounded by $3\tau$.

Thus, we obtain the relationship between $Q_{\text{LP-Del-N}}(G, \tau)$ and $Q_{\text{Del-N}}(G, \tau)$:

\begin{lemma}
\label{lem:lp_bound}
For any $G$ and any $\tau$,
\begin{equation}
\label{e:lp_bound}
   |Q_{\text{Del-N}}(G, 3\tau)| \leq 3 \cdot |Q_{\text{LP-Del-N}}(G, \tau)|. 
\end{equation}
\end{lemma}

\begin{proof}
    To begin with, with the optimal solution of $Q_{\text{LP-Del-N}}(G, \tau)$, we can construct a subgraph of $G$ with a degree upper bound of $3\tau$ by Algorithm~\ref{alg:lp_clip_graph}.

\begin{algorithm}[htbp]
\caption{$\tau$-degree Subgraph Construction with LP}
\label{alg:lp_clip_graph}
\SetAlgoLined
\LinesNumbered
\SetNoFillComment 
\DontPrintSemicolon
\KwIn{Graph \(G= (V, E)\), degree bound \(\tau\)}
\KwOut{A subgraph $G^*$}

$ \{x_v\}, \{y_e\} \leftarrow Q_{\text{LP-Del-N}}(G, \tau)$\;

$G^* \gets G$ 

\For{each node \(v \in G^*\)}{
    \If{\(x_v > 1/3\)}{
    \tcc{Remove \(v\) and incident edges from \(G^*\)}
        $G^* \gets G^* \setminus \{v\}$ 
    }
}

\Return $ G^*$  
\end{algorithm}

    Given $G$ and $\tau$, let the output from Algorithm~\ref{alg:lp_clip_graph} be $G^*=(V^*, E^*)$
    An edge $e = (v, v')$ is retained in $G^*$ only if both $v$ and $v'$ is not removed from $G$, i.e., $x_v \leq 1/3$ and $x_{v'} \leq 1/3$. 
For such edges $e$, $y_e \geq 1 - x_v - x_{v'} \geq 1/3$. By the definition of the LP, we have 
\begin{equation*}
    \sum_{e \in E(v)} y_e \leq \tau
\end{equation*}
for any node $v \in V$.
Assign $d_e = 1$ for all retained edges $e$. Since $d_e \leq 3 \cdot y_e$, the degree of any node $v$ in $G^*$ is bounded by
\begin{equation*}
    \sum_{e \in E^*(v)} d_e \leq 3 \cdot \sum_{e \in E(v)} y_e \leq 3\tau. 
\end{equation*}

Furthermore, to construct $G^*$, a node $v$ is removed if $x_v > 1/3$, otherwise $v$ is retained. 
We define $n_v = 1$ if $v$ is clipped and $n_v = 0$ otherwise. It follows that $n_v \leq 3 \cdot x_v$.
Therefore, the total number of node removals from $G$ to obtain $G^*$ is
\begin{equation*}
    \sum_{v \in V} n_v \leq 3 \cdot \sum_{v \in V}  x_v = 3 \cdot |Q_{\text{LP-Del-N}}(G, \tau)|.
\end{equation*}
By definition, $|Q_{\text{Del-N}}(G, 3\tau)|$ represents the minimum node removal to bound the degree by $3\tau$. Therefore, we can conclude that
\begin{equation*}
     |Q_{\text{Del-N}}(G, 3\tau)| \leq 3 \cdot |Q_{\text{LP-Del-N}}(G, \tau)|. \qedhere
\end{equation*}
\end{proof}

Therefore, by feeding $Q_{\text{LP-Del-N}}(G, \tau)$ into the SVT, we obtain a $\tau$ such that $\tau \leq 2 \deg(G)$ and $|Q_{\text{LP-Del-N}}(G, \tau)|$ is bounded by $\tilde{O}(1 / \varepsilon)$.
Following a similar analysis to Section~\ref{subsubsec:method-deg-nodedp_exp}, we modify Line~\ref{line:tau-postprocess} of Algorithm~\ref{alg:exp_node_max_deg_approx} to compute $\tau^*$ as:
\begin{equation}
\label{e:tau*}
\tau^* \gets 3\tau + 3 \cdot |Q_{\text{LP-Del-N}}(G,\tau)| 
+ \mathrm{Lap}(\frac{6}{\varepsilon}) +
\frac{6}{\varepsilon}\ln\max(\frac{1}{\delta}, \frac{2}{\beta})+ 1,
\end{equation}
which yields an $\tau^*$ that satisfies the three desired properties:

\begin{lemma}
\label{lem:poly_bound}
Given any $\varepsilon$, $\beta$, and $\delta$, for any graph $G$, the polynomial-time node-DP maximum degree approximation algorithm returns a $\tau^*$ such that with probability at least $1-\beta$,
\begin{equation*}
    \tau^* \leq 6\mathrm{deg}(G) + \frac{24}{\varepsilon}\ln\log(4\mathrm{deg}(G)) + \frac{48}{\varepsilon}\ln\frac{4}{\beta} + \frac{12}{\varepsilon}\ln\max(\frac{1}{\delta}, \frac{2}{\beta})+ 1,
\end{equation*}
and 
\begin{equation*}
    N_{\tau^*}(G) \leq \frac{24}{\varepsilon}\ln\log(4\mathrm{deg}(G)) + \frac{48}{\varepsilon}\ln\frac{4}{\beta}.
\end{equation*}
Furthermore, with probability at least $1-\delta$, 
\begin{equation*}
    \tau^*+ N_{\tau^*}(G) \leq 2\tau^*. 
\end{equation*}
\end{lemma}

\begin{proof}
Similar to the proof of Lemma~\ref{lem:exp_clip_node_bound}, the output $\tau$ from the SVT in the node-DP polynomial-time maximum degree approximation satisfies $\tau \leq 2\mathrm{deg}(G)$ and 
\begin{equation*}
|Q_{\text{LP-Del-N}}(G, \tau)| \leq \frac{8}{\varepsilon} \ln \log(4\deg(G)) + \frac{16}{\varepsilon} \ln \frac{4}{\beta}
\end{equation*}
with probability at least $1 - \beta/2$.

Next, $\tau$ is post-processed to $\tau^*$ as defined in (\ref{e:tau*}). Similarly to Lemma~\ref{lem:exp_clip_node_bound}, with probability at least $1 - min(\beta/2, \delta)$,
\begin{equation*}
    \tau^* \leq 3\tau + 3|Q_{\text{LP-Del-N}}(G, \tau)| + 1 + \frac{12}{\varepsilon}\ln\max(\frac{1}{\delta}, \frac{2}{\beta}), 
\end{equation*}
and
\begin{equation*}
    \tau^* \geq 3\tau + 3 \cdot |Q_{\text{LP-Del-N}}(G, \tau)| + 1.
\end{equation*}

By Lemma~\ref{lem:lp_bound}, this directly implies that 
\begin{equation*}
     \tau^* \geq 3\tau + |Q_{\text{Del-N}}(G, 3\tau)| + 1. 
\end{equation*}
Applying (\ref{e:ntau_bound}) with $\tau = 3\tau$ and using Lemma~\ref{lem:lp_bound}, we obtain
\begin{equation*}
     N_{\tau^*}(G) \leq |Q_{\text{Del-N}}(G, 3\tau)| \leq 3 \cdot |Q_{\text{LP-Del-N}}(G, \tau)|.
\end{equation*}
    
By applying the union bound and incorporating the upper bound on $\tau$ and $|Q_{\text{LP-Del-N}}(G, \tau)|$, we conclude that, with probability at least $1 - \beta$, the bounds on $\tau^*$ and $N_{\tau^*}(G)$ hold as claimed.

Finally, $N_{\tau^*}(G) \leq 3 \cdot |Q_{\text{LP-Del-N}}(G, \tau)|$ holds with probability at least $1 - \delta$. Since $\tau^*$ includes the term $3 \cdot |Q_{\text{LP-Del-N}}(G, \tau)|$, it follows directly that 
 \begin{equation*}
     \tau^* + N_{\tau^*}(G) \leq 2\tau^*. \qedhere
 \end{equation*}
\end{proof}

\red{The revised algorithm maintains the same $\varepsilon$-node-DP guarantee as Algorithm~\ref{alg:exp_node_max_deg_approx}.}
\orange{
For utility, the allocation of error failure probabilities remains unchanged, 
and the upper bounds on $\tau^*$ and $N_{\tau^*}(G)$ increase by roughly a factor of $3$ under the same parameters.
}
In terms of runtime, since each invocation of $Q_{\text{LP-Del-N}}(G, \tau)$ runs in polynomial time and there are $O(\log \deg(G))$ iterations, the overall algorithm also runs in polynomial time.

\subsection{Combined Solution}
\label{subsec:method-combine}

By combining the distance-preserving clipping mechanism from Section~\ref{subsec:method-clipping} with the node-DP polynomial-time maximum degree approximation algorithm from Section~\ref{subsec:method-deg_approx}, we develop the N2E framework that reduces an arbitrary node-DP graph analytical task to an edge-DP one as follows. 
Given a graph analytical task $Q$, an edge-DP mechanism $\mathcal{M}_Q^{\text{edge}}$, an input graph $G$, privacy budgets $\varepsilon$ and $\delta$, and error failure probability $\beta$, the framework proceeds in three steps:
\begin{itemize}
    \item Compute $\tau^*$ using the node-DP polynomial-time maximum degree approximation on input graph $G$, with privacy budgets $2\varepsilon/3$ and $\delta/2$, and error failure probability $2\beta/3$.
    
    \item Clip graph $G$ using the distance-preserving clipping mechanism described in Algorithm~\ref{alg:clip_graph}, with the degree upper bound set to $\tau^*$.
    Denote the resulting graph as $\overline{G}$.
    
    \item Apply the edge-DP mechanism $\mathcal{M}_Q^{\text{edge}}$ to the clipped graph $\overline{G}$ with privacy budgets $\varepsilon/6\tau^*$ and $\delta/4\tau^*$, and error failure probability $\beta/3$.
\end{itemize}
The framework outputs the result of the edge-DP mechanism, denoted as $\mathcal{M}_Q^{\text{node}}(G, \varepsilon, \delta, \beta)$.
For privacy, 

\begin{theorem}
\label{the:n2e_privacy}
The N2E framework satisfies $(\varepsilon, \delta)$-node-DP.
\end{theorem}

\begin{proof}
In the first step of N2E, the polynomial-time maximum degree approximation algorithm ensures $2\varepsilon/3$-node-DP to compute the approximation $\tau^*$,
and guarantees an edge distance bound of \( 2\tau^* \) between neighboring graphs after clipping.  
This bound holds with probability at least $1 - \delta/2$.

The second step of N2E, graph clipping, does not consume any privacy budget.

In the third step, by group privacy (Lemma~\ref{lem:gp}), the edge-DP mechanism with $\varepsilon' = \varepsilon / 6\tau^*$ and $\delta' = \delta / 4\tau^*$ achieves $(\varepsilon/3, \delta/2)$-node-DP when the edge distance bound holds. 
Since this bound fails with probability at most $\delta/2$, the mechanism satisfies $(\varepsilon/3, \delta)$-node-DP.

By sequential composition (Lemma~\ref{lem:bc}), combining all three steps yields an overall $(\varepsilon, \delta)$-node-DP for the N2E framework.
\end{proof}

For utility, we have:

 \begin{theorem}
\label{the:n2e_util}

Given any $\varepsilon$, $\beta$, and $\delta$, for any graph $G$, task $Q$, and edge-DP mechanism $\mathcal{M}_Q^{\text{edge}}$, with probability at least $1-\beta$, the N2E framework outputs a node-DP result with error 
\begin{equation*}
       |Q(G) - Q(\overline{G})| + \mathrm{Err}_Q^{\text{\emph{edge}}}(\overline{G},\varepsilon/6\tau^*, \delta/4\tau^*,\beta/3),
\end{equation*}
where $\overline{G}$ is obtained by clipping edges of $O(\log(\log(\mathrm{deg}(G))/ \beta)/\varepsilon)$ nodes from $G$, \\
and $\mathrm{Err}_Q^{\text{\emph{edge}}}(\overline{G},\varepsilon/6\tau^*, \delta/4\tau^*,\beta/3)$ denotes the error of the edge-DP mechanism on task $Q$ given clipped graph $\overline{G}$, privacy budgets $\varepsilon/6\tau^*$ and $\delta/4\tau^*$, and error failure probability $\beta/3$.

\end{theorem}

\begin{proof}
In the worst case, all edges incident to nodes with degree exceeding $\tau^*$ are removed by the distance-preserving clipping mechanism. By Lemma~\ref{lem:poly_bound}, with probability at least $1-2\beta/3$, $N_{\tau^*}(G)$ is bounded by $O(\log(\log(\mathrm{deg}(G))/ \beta)/\varepsilon)$. 
Hence, $\overline{G}$ is obtained by clipping edges of at most $O(\log(\log(\mathrm{deg}(G))/ \beta)/\varepsilon)$ nodes from $G$ under the same probability.

The edge-DP mechanism error $\mathrm{Err}_Q^{\text{{edge}}}(\overline{G}, \varepsilon/6\tau^*, \delta/4\tau^*, \beta/3)$ holds with probability at least $1 - \beta/3$. By the union bound, the total error of the N2E output holds with probability at least $1 - \beta$, completing the proof.
\end{proof}

Specifically, the first component of the error, $|Q(G) - Q(\overline{G})|$, represents the clipping bias introduced by the difference between the query result on the original graph $G$ and the clipped graph $\overline{G}$.
\orange{For example, when $Q$ is the edge count, this bias is exactly the number of clipped edges.}
\blue{
By Theorem~\ref{the:n2e_util}, this component is at most $\Delta Q^{(\rho-1)}$, where $\rho = \tilde{O}(1/\varepsilon)$ with constant probability. 
This implies that only a negligible portion of the graph is affected by clipping. 
Moreover, according to Lemma~\ref{lem:dno}, this part of the error is optimal.
It follows that any structural impact on the graph (e.g., connectivity) introduced by the distance-preserving clipping mechanism is both unavoidable and minimal within the N2E framework.
}




The second component of the error, $\mathrm{Err}_Q^{\text{edge}}(\overline{G}, \varepsilon / 6\tau^*, \delta/4\tau^*, \beta/3)$, captures the error introduced by the edge-DP mechanism. In practice, this error approximates 
$\mathrm{Err}_Q^{\text{edge}}(G,\varepsilon/6\tau^*, \delta/4\tau^*, \beta/3)$, since many edge-DP mechanisms for common graph analytical tasks exhibit similar or smaller error on the clipped graph $\overline{G}$ compared to $G$.
Moreover, by Lemma~\ref{lem:poly_bound} that gives the upper bound of $\tau^*$, this error amplifies the original edge-DP error of $\mathrm{Err}_Q^{\text{edge}}(G, \varepsilon, \delta, \beta)$ by only a factor of $\tilde{O}(\mathrm{deg}(G))$, which achieves our desired utility.

\section{Application}
\label{sec:app}
In this section, we apply our N2E framework to three common graph analytical tasks: edge counting, maximum degree estimation, and degree distribution under node-DP.
We analyze the error for each task individually. 
For convenience, all results stated in this section hold with constant probability.




\paragraph{Edge counting}

Under $\varepsilon$-edge-DP, edge counting using the Laplace mechanism incurs an error of $O(1/\varepsilon)$.  
We integrate this into our N2E framework as the edge-DP mechanism.
For the clipping bias, since each node contributes at most $\mathrm{deg}(G)$ edges, the total number of edges clipped to obtain $\overline{G}$ is bounded by $O(\mathrm{deg}(G)\log\log\mathrm{deg}(G)/\varepsilon)$.
The error from the edge-DP mechanism is bounded by $O(\mathrm{deg}(G)/\varepsilon + \log\log\mathrm{deg}(G)/\varepsilon^2 + \log(1/\delta)/\varepsilon^2)$.
As a result, the total error is bounded by
$\tilde{O}(\mathrm{deg}(G)/\varepsilon + 1/\varepsilon^2)$,
which matches our target of upgrading the error of edge-DP by only a factor of $\tilde{O}(\deg(G))$ with a small additive term of $\tilde{O}(1/\varepsilon)$.



\paragraph{Maximum degree estimation.}

For $\varepsilon$-edge-DP, the Laplace mechanism incurs an error of $O(1/\varepsilon)$.  
However, when applied within our N2E framework, this mechanism introduces an error of $\tilde{O}(\mathrm{deg}(G)/\varepsilon)$, which has limited utility.
To address this, we use Algorithm~\ref{alg:edge_dp_max_deg_esti} as the edge-DP mechanism in our N2E framework, which incurs an edge distance of $O(\log\mathrm{deg}(G)/\varepsilon)$ under $\varepsilon$-edge-DP. 
We also evaluate the edge distance instead of the exact error for the N2E framework output.
For the clipping bias, the edge distance between $G$ and $\overline{G}$ is bounded by $O(\mathrm{deg}(G)\log\log\mathrm{deg}(G)/\varepsilon)$.
For the edge-DP mechanism application, the edge distance is bounded by $O(\mathrm{deg}(G)\log\mathrm{deg}(G)/\varepsilon + \log\mathrm{deg}(G)\log\log\mathrm{deg}(G)/\varepsilon^2 + \log\mathrm{deg}(G)\log(1/\delta)/\varepsilon^2)$.
Therefore, the total edge distance is asymptotically dominated by the edge-DP mechanism application and is bounded by $\tilde{O}(\mathrm{deg}(G)/\varepsilon + 1/ \varepsilon^2)$, achieving our target of upgrading the edge-DP error by $\tilde{O}(\mathrm{deg}(G))$ with an additive term of $\tilde{O}(1/\varepsilon)$.

\paragraph{Degree distribution}
A straightforward solution for this task is to integrate the Laplace mechanism into N2E, which introduces a clipping bias of $\tilde{O}(\mathrm{deg}(G)/\varepsilon)$ and an error of $\tilde{O}(\mathrm{deg}^{1.5}(G)/\varepsilon + \mathrm{deg}^{0.5}(G)/\varepsilon^2)$ introduced by the Laplace noise. 
This approach satisfies $(\varepsilon, \delta)$-node-DP and achieves our goal of amplifying the edge-DP error of $O(\mathrm{deg}(G)^{0.5}/\varepsilon)$ by a factor of $\tilde{O}(\mathrm{deg}(G))$, with an additive term of $\tilde{O}(1/\varepsilon)$. 
Notably, in the node-DP setting, the existing approach by~\citet{day2016publishing} employs a graph clipping mechanism $\pi_{\theta}(G)$ that limits the global sensitivity of the degree histogram to $2\theta + 1$ with a degree bound $\theta$.
We observe that combining our node-DP maximum degree approximation to compute $\theta$ with their clipping mechanism yields a pure $\varepsilon$-node-DP solution with asymptotically the same utility guarantees.
Thus, we employ this approach in our work.
Since each clipped edge affects $2$ node degrees, the resulting clipping bias is bounded by $O(\mathrm{deg}(G)\log\log\mathrm{deg}(G)/\varepsilon)$.
The error introduced by the Laplace noise is bounded by $O(\mathrm{deg}^{1.5}(G)/\varepsilon + \mathrm{deg}^{0.5}(G)\log\log\mathrm{deg}(G)/\varepsilon^2 + \mathrm{deg}^{0.5}(G)\log(1/\delta)/\varepsilon^2)$. 
In summary, the total error is bounded by $\tilde{O}(\mathrm{deg}^{1.5}(G)/\varepsilon + \mathrm{deg}^{0.5}(G)/\varepsilon^2)$, which meets our goal of amplifying the edge-DP error by $\tilde{O}(\mathrm{deg}(G))$ with an additive term of $\tilde{O}(1/\varepsilon)$.

\paragraph{Limitation of the N2E framework}
N2E is applicable to any graph analytical tasks with edge-DP solutions, either being $\varepsilon$-edge-DP or being $(\varepsilon, \delta)$-edge-DP.
However, N2E is not suitable for tasks that lack corresponding edge-DP solutions or for simple node-DP tasks with trivial solutions, such as node counting. 
Moreover, the performance of N2E heavily depends on the underlying edge-DP mechanism. 
If the edge-DP mechanism has poor utility or high computational cost, these limitations will carry over to N2E.

\section{Implementation}
\label{sec:implement}

\begin{figure}[htbp]
  \centering
  \includegraphics[width=0.7\linewidth]{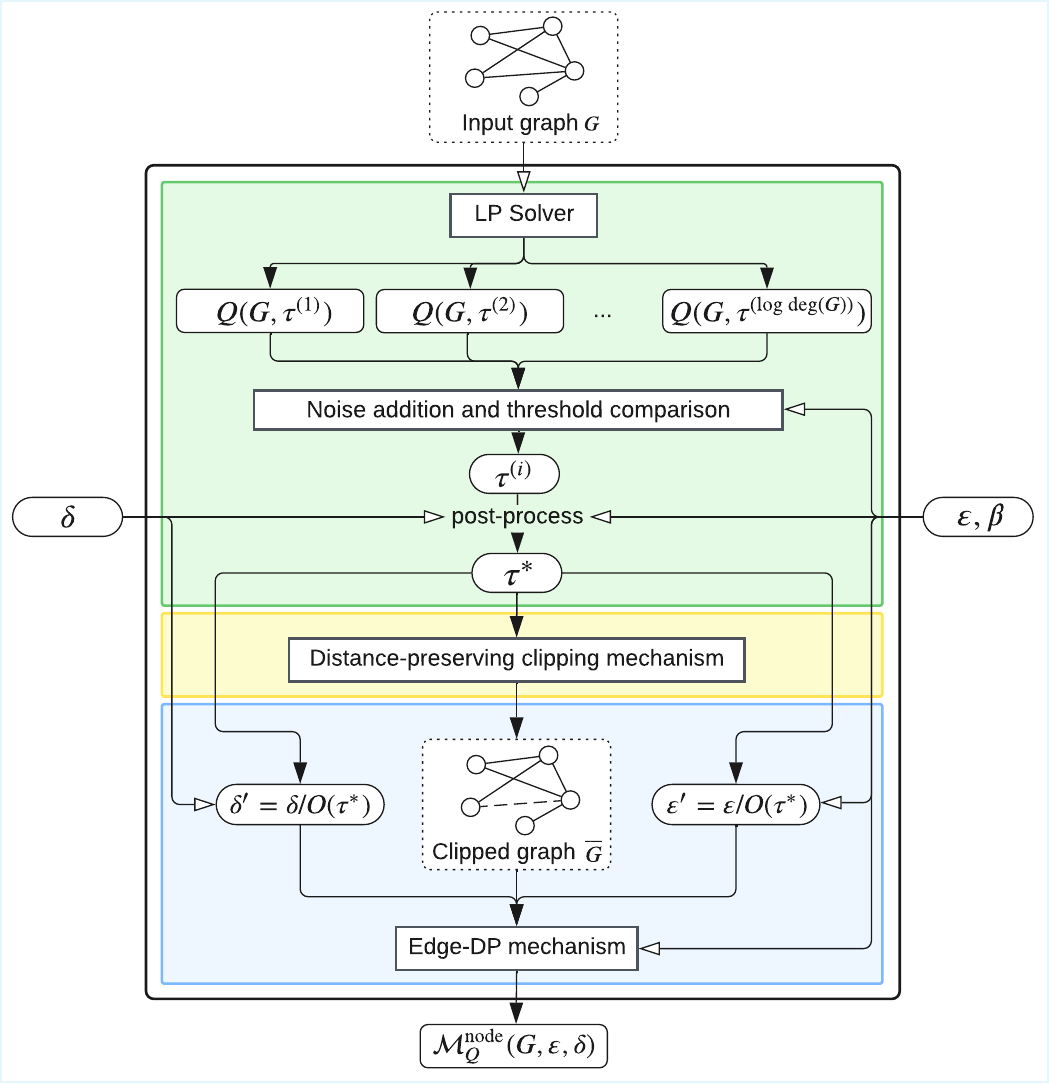}
  \caption{Overview of the N2E framework.}
  \label{fig:workflow}
\end{figure}

We implemented our N2E framework on top of CPLEX~\cite{manual1987ibm}.  
An overview is shown in Figure~\ref{fig:workflow}, where the three main steps are highlighted in different colors.
Specifically, the input graph $G$ is utilized in both the LP construction and the clipping mechanism. 
Regarding input parameters, the framework includes three procedures that consume $\varepsilon$ and $\beta$, two of which are involved in the maximum degree approximation.
The parameter $\delta$ is used in the post-processing step to compute $\tau^*$ and in the edge-DP mechanism.

\paragraph{Runtime optimization.}
In practice, the first step of the N2E framework, maximum degree approximation, accounts for the majority of the runtime in our experiments.  
Although this procedure achieves polynomial-time complexity with LP, solving the LP remains computationally expensive.
This is because the number of variables in the LP, which is equal to the sum of nodes and edges, can be substantial for real-world datasets.  
This issue is further amplified as one LP is to be solved in each SVT iteration, and the SVT typically takes multiple iterations to identify the desired $\tau$. 

To reduce sequential runtime, we first parallelized the LPs in SVT iterations within the maximum degree approximation, as illustrated in Figure~\ref{fig:workflow}. 
The LP result is $0$ for $\tau \geq \mathrm{deg}(G)$.
To avoid unnecessary computation, we implement early termination. 
If any process returns a qualified $\tau$ that exceeds the SVT threshold and all smaller candidates are unqualified, the process pool is terminated early and returns the corresponding $\tau$.

To further reduce the runtime for each candidate $\tau$, we leverage the dual simplex method provided by the LP solver CPLEX, which solves LPs by moving towards the optimal solution from above. 
Since we are only interested in $\tau$ candidates whose LP result exceeds the SVT threshold, we stop the computation early if the current LP upper bound of $\tau$ already falls below the threshold.  
This strategy significantly reduces runtime, especially in cases where resource constraints prevent parallelization.
\section{Experiment}
\label{sec:experiment}
We conduct experiments on the three graph analytical tasks under node-DP discussed in Section~\ref{sec:app}, along with an additional task of two-line path counting, and compare the performance of our N2E framework with state-of-the-art node-DP algorithms. 
\reviseagain{ 
For edge counting, we compare with Gem-EC~\cite{raskhodnikova2016lipschitz}, R2T~\cite{dong2022r2t}, and OPT$^2$~\cite{dong2024instance}.
For maximum-degree estimation, we compare with N-Count, a naive baseline that releases the node count with
Laplace noise.
}
\reviseagain{
For two-line path counting, we compare with R2T and OPT$^2$.
}
For degree distribution, we compare our results with \reviseagain{Gem-Histogram~\cite{raskhodnikova2016lipschitz} and} DLL-Histogram~\cite{day2016publishing}.

\subsection{Setup}
\label{subsec:experiment-setup}

\begin{table}[htb!]
\centering
\setlength{\tabcolsep}{6pt}
\resizebox{0.65\columnwidth}{!}{%
\begin{tabular}{l||r|r|r|r}
\toprule
\textbf{Dataset}      & \textbf{Nodes} & \textbf{Edges}  & \textbf{Max Degree} & \blue{\textbf{Connectivity}} \\
\midrule\midrule
EUEmail               & 265,009        & 364,481         & 7,636               & \blue{0.21} \\
Deezer                & 143,884        & 846,915         & 420                 & \blue{0.17} \\
Amazon1               & 262,111        & 899,792         & 420                 & \blue{0.08} \\
Amazon2               & 334,863        & 925,872         & 549                 & \blue{0.06} \\
DBLP                  & 317,080        & 1,049,866       & 343                 & \blue{0.18} \\
Amazon0601            & 403,394        & 2,443,408       & 2,752               & \blue{0.10} \\
Web-Google            & 875,713        & 4,322,051       & 6,332               & \blue{0.01} \\
\bottomrule
\end{tabular}%
}
\caption{Summary of graph datasets used in the experiment. Connectivity is computed as the ratio of intra-cluster to inter-cluster edges based on Louvain clustering~\cite{blondel2008fast}.}
\label{tab:dataset}
\end{table}


\paragraph{Datasets}
Our experiments use seven real-world graph datasets: EUEmail, Deezer, Amazon1, Amazon2, DBLP, Amazon0601, and Web-Google, all obtained from SNAP~\cite{leskovec2014snap}. 
These datasets span various domains, including e-commerce, academic collaboration, email communication, and web networks.  
Details of all datasets are provided in Table~\ref{tab:dataset}.  
All graphs are preprocessed to be undirected for consistency.
For the Deezer and Amazon1 datasets, we follow the same preprocessing procedure as in~\cite{dong2022r2t} to match their experimental results.

\paragraph{Queries}
We use four node-DP graph analysis queries: edge counting ($Q_{\mathrm{EC}}$), two-line path counting ($Q_{\mathrm{TP}}$), maximum degree estimation ($Q_{\mathrm{MD}}$), and degree distribution ($Q_{\mathrm{DD}}$).

For $Q_{\mathrm{EC}}$ and $Q_{\mathrm{TP}}$, we use additive error as the evaluation metric.  
For $Q_{\mathrm{MD}}$, we use the rank error\footnote{\red{
For results exceeding the true value, we compute the rank error as the absolute error.
}
}
for evaluation.
For $Q_{\mathrm{DD}}$, we construct the degree histogram on a logarithmic scale of degrees \reviseagain{(with Gem-Histogram adapted accordingly)} and report the $\ell_1$ error for evaluation.


\paragraph{Experimental parameters}
All experiments are conducted on a Linux server equipped with dual Intel Xeon 22-core 2.20~GHz CPUs and 1~TB of memory.
Each experiment is repeated for $10$ rounds, and we report the average result after excluding the top $2$ and bottom $2$ runs to mitigate the randomness of the DP mechanism. 
All algorithms are limited to a $6$-hour runtime; any algorithm exceeding this limit is terminated.

We set $\beta = 0.1$ and $\delta = 2^{-30}$ in all experiments.  
By default, we set $\varepsilon = 0.8$ for $Q_{\mathrm{EC}}$, $Q_{\mathrm{TP}}$ and $Q_{\mathrm{MD}}$.  
Since the degree histogram involves multiple queries, we relax the default $\varepsilon$ to $3.2$ for $Q_{\mathrm{DD}}$.
The division of $\varepsilon$, $\delta$, and $\beta$ across the procedures in Figure~\ref{fig:workflow} is flexible and does not asymptotically affect the theoretical guarantees.  
Based on empirical results, we allocate $(\varepsilon, \delta, \beta)$ as $(20\%, 0\%, 20\%)$, $(20\%, 100\%, 0.01\%)$, and $(60\%, 0\%, 79.99\%)$ to the first, second, and third procedures, respectively.
\reviseagain{ 
For baselines Gem-EC and Gem-Histogram that first use the generalized exponential mechanism for choosing a degree threshold $D$, we allocate $(40\%, 99.99\%)$ and $(60\%, 0.01\%)$ for $(\varepsilon, \beta)$ on selecting $D$ and running the node-DP mechanism for the final result with selected $D$, respectively.
} 

\reviseagain{
 For Gem-Histogram, we use the Primal-Dual Hybrid Gradient (PDHG) algorithm to solve the $\Phi$-optimal flow problem defined in Definition~4.2 of~\cite{raskhodnikova2016lipschitz} for every candidate $D$.
We declare convergence when the relative primal-dual gap is at most $10^{-7}$. 
Theoretical analysis in~\cite{raskhodnikova2016lipschitz} notes that numerical approximation may slightly increase sensitivity and should therefore be compensated by a corresponding noise margin; for simplicity, our experiments omit this correction and use the nominal noise scale.
}
Both R2T and DLL-Histogram rely on a degree upper bound, which is set to the smallest power of $2$ exceeding the number of nodes for each dataset. 
For example, the bound is set to $2^{18} = 262{,}144$ for Deezer.
\reviseagain{
For all gem-based threshold selections, $D$ is chosen from the powers of $2$ between $1$ and the largest power of $2$ less than $N$.
}


\subsection{Experiment Results}
\begin{table*}[!t]
\centering
\renewcommand{\arraystretch}{1.1}
\setlength{\tabcolsep}{6pt}
\resizebox{1\textwidth}{!}{
\begin{tabular}{c|l|c||c|c|c|c|c|c|c}
\toprule
\multicolumn{3}{c||}{\textbf{Dataset}} 
& \textbf{EUEmail}
& \textbf{Deezer}
& \textbf{Amazon1}
& \textbf{Amazon2}
& \textbf{DBLP}
& \textbf{Amazon0601}
& \textbf{Web-Google} \\
\midrule\midrule

\multirow{14}{*}{$Q_{\mathrm{EC}}$} 
  & \multicolumn{2}{c||}{Query result} 
  & $3.6 \times 10^5$
  & $8.5 \times 10^5$
  & $9.0 \times 10^5$ 
  & $9.3 \times 10^5$ 
  & $1.0 \times 10^6$
  & $2.4 \times 10^6$ 
  & $4.3 \times 10^6$ \\
\cmidrule(lr){2-10}
  & \multirow{3}{*}{N2E} 
  & Rel. err (\%) 
  & \cellcolor{gray!35}6.20 & \cellcolor{gray!35}0.18 & \cellcolor{gray!35}0.12 & 0.19 & 0.23 & \cellcolor{gray!35}0.24 & \cellcolor{gray!35}0.76 \\
  & & Err. std (\%) 
  & \orange{0.77} & \orange{0.07} & \orange{0.07} & \orange{0.17} & \orange{0.14} & \orange{0.09} & \orange{0.20} \\
  & & Time (s) 
  & 279.85 & 66.39 & 65.41 & 81.66 & 71.65 & 267.66 & 1363.91 \\
  \cmidrule(lr){2-10}
  & \multirow{3}{*}{\reviseagain{Gem-EC}}
  & Rel. err. (\%)
  & \reviseagain{10.75} & \reviseagain{0.37} & \reviseagain{0.37} & \cellcolor{gray!35}\reviseagain{0.15} 
  & \cellcolor{gray!35}\reviseagain{0.22} & \reviseagain{0.42} & \reviseagain{0.78} \\
  & & Err. std (\%)
  & \reviseagain{2.88} & \reviseagain{0.19} & \reviseagain{0.28} & \reviseagain{0.09} & \reviseagain{0.25} & \reviseagain{0.31} & \reviseagain{1.04} \\
  & & Time (s)
  & \cellcolor{gray!12}\reviseagain{1.77} 
  & \cellcolor{gray!12}\reviseagain{6.23} 
  & \cellcolor{gray!12}\reviseagain{9.25} 
  & \cellcolor{gray!12}\reviseagain{5.73} 
  & \cellcolor{gray!12}\reviseagain{6.41} & \cellcolor{gray!12}\reviseagain{22.88} 
  & \cellcolor{gray!12}\reviseagain{37.45} \\

\cmidrule(lr){2-10}
  & \multirow{3}{*}{OPT$^2$} 
  & Rel. err (\%) 
  & 10.51 & 0.47 & 0.35 & 0.23 & 0.45 & 0.40 & 1.73 \\
  & & Err. std (\%) 
  & \orange{1.37} & \orange{0.05} & \orange{0.02} & \orange{0.02} & \orange{0.02} & \orange{0.05} & \orange{0.94} \\
  & & Time (s) 
  & 76.46 & 170.87 & 180.86 & 141.66 & 181.56 & 521.03 & 1370.32 \\
\cmidrule(lr){2-10}
  & \multirow{3}{*}{\red{R2T}} 
  & Rel. err (\%) 
  & \red{23.32} & \red{0.98} & \red{0.87} & \red{0.75} & \red{0.71} & \red{1.72} & \red{5.68} \\
  & & Err. std (\%) 
  & \orange{1.02} & \orange{0.16} & \orange{0.13} & \orange{0.04} & \orange{0.03} & \orange{0.18} & \orange{0.38} \\
  & & Time (s) 
  & \red{3.32} & \red{10.43} & \red{10.70} & \red{9.94} & \red{11.78} & \red{26.38} & \red{55.86} \\

\midrule\midrule

\multirow{11}{*}{$Q_{\mathrm{TP}}$} 
  & \multicolumn{2}{c||}{Query result} 
  & $2.0 \times 10^8$
  & $2.2 \times 10^7$
  & $9.1 \times 10^6$ 
  & $9.8 \times 10^6$ 
  & $2.2 \times 10^7$ 
  & $7.2 \times 10^7$ 
  & $7.3 \times 10^8$ \\
\cmidrule(lr){2-10}
  & \multirow{3}{*}{N2E} 
  & Rel. err (\%) 
  & \cellcolor{gray!35}43.13 & \cellcolor{gray!35}3.16 & \cellcolor{gray!35}11.34 & 9.52 & 7.43 & \cellcolor{gray!35}10.71 & \cellcolor{gray!35}9.33 \\
  & & Err. std (\%) 
  & \orange{2.84} & \orange{1.54} & \orange{3.51} & \orange{2.58} & \orange{3.09} & \orange{4.16} & \orange{3.48} \\
  & & Time (s) 
  & \cellcolor{gray!12}269.00 & \cellcolor{gray!12}66.56 & \cellcolor{gray!12}65.72 & \cellcolor{gray!12}79.59 & \cellcolor{gray!12}68.88 & \cellcolor{gray!12}270.52 & \cellcolor{gray!12}1358.04 \\
\cmidrule(lr){2-10}
  & \multirow{3}{*}{OPT$^2$}  & Rel. err (\%) 
  & \multicolumn{3}{c|}{\multirow{3}{*}{\centering\textit{Over time limit}}} & \cellcolor{gray!35}4.91 & \multicolumn{3}{c}{\multirow{3}{*}{\centering\textit{Over time limit}}} \\
  &         & Err. std (\%) 
  & \multicolumn{3}{c|}{} & \orange{0.20} & \multicolumn{3}{c}{} \\
  &         & Time (s) 
  & \multicolumn{3}{c|}{} & 12498.95 & \multicolumn{3}{c}{} \\
\cmidrule(lr){2-10}
  & \multirow{3}{*}{\red{R2T}}  & Rel. err (\%) 
  & \multicolumn{1}{c|}{\multirow{3}{*}{\textit{Out of memory}}} & \red{4.94} & \red{11.97} & \red{8.92} & \cellcolor{gray!35}\red{6.22} & \red{24.35} & \multicolumn{1}{c}{\multirow{3}{*}{\textit{Out of memory}}} \\
  &     & Err. std (\%) 
  &     & \orange{1.04} & \orange{0.31} & \orange{0.19} & \orange{0.82} & \orange{1.88} & \\
  &     & Time (s) 
  &     & \red{280.20} & \red{113.80} & \red{122.51} & \red{284.82} & \red{953.21} & \\
\midrule\midrule

\multirow{8}{*}{$Q_{\mathrm{MD}}$} 
  & \multicolumn{2}{c||}{Query result} 
  & 7,636 & 420 & 420 & 549 & 343 & 2,752 & 6,332 \\
\cmidrule(lr){2-10}
  & \multirow{3}{*}{N2E} 
  & Rel. err (\%) 
  &  0.37 &  3.10 &  1.87 &  1.12 &  1.10 &  0.65 & 0.18 \\
  & & Err. std (\%) 
  & \orange{0.12} & \orange{0.40} & \orange{0.24} & \orange{0.33} & \orange{0.09} & \orange{0.08} & \orange{0.09} \\
  & & Time (s) 
  & 269.22 & 66.89 & 65.99 & 80.09 & 70.22 & 271.99 & 1359.70 \\



\cmidrule(lr){2-10}
  & \multirow{3}{*}{\red{N-Count}} 
  & Rel. err (\%) 
  & \red{3370.52} & \red{34158.06} & \red{62307.38} & \red{60894.99} & \red{92343.05} & \red{14558.22} & \red{13729.96} \\
  & & Err. std (\%) 
  & \orange{0.01} & \orange{0.10} & 
  \orange{0.15} & \orange{0.10} & \orange{0.24} & \orange{0.04} & 
  \orange{0.00} \\
    & 
  & Time (s) 
  &  - 
  &  - 
  &  - 
  &  - 
  &  - 
  &  - 
  &  - \\
\midrule\midrule

\multirow{11}{*}{$Q_{\mathrm{DD}}$} 
  & \multicolumn{2}{c||}{Query result} 
  & $2.6 \times 10^5$ 
  & $1.4 \times 10^5$ 
  & $2.6 \times 10^5$ 
  & $3.4 \times 10^5$
  & $3.2 \times 10^5$
  & $4.0 \times 10^5$
  & $8.8 \times 10^5$ \\
\cmidrule(lr){2-10}
  & \multirow{3}{*}{N2E} 
  & Rel. err (\%) 
  & \cellcolor{gray!35}8.17 & 3.38 & 1.21 & 1.05 & \cellcolor{gray!35}1.94 & 3.75 & \cellcolor{gray!35}8.76 \\
  & & Err. std (\%) 
  & \orange{1.20} & \orange{0.93} & \orange{0.11} & \orange{0.32} & \orange{0.29} & \orange{0.36} & \orange{1.79} \\
  & & Time (s) 
  & 30.10 & 63.67 & 64.39 & 65.95 & 72.68 & 153.51 & 272.08 \\
  

  \cmidrule(lr){2-10}
  & \multirow{3}{*}{\reviseagain{Gem-Histogram}}
  & Rel. err. (\%)
  & \reviseagain{19.86} & 
  \cellcolor{gray!35}\reviseagain{3.29} 
  & \cellcolor{gray!35}\reviseagain{1.15}
  & \cellcolor{gray!35}\reviseagain{0.88} 
  & \reviseagain{2.14} & 
  \cellcolor{gray!35}\reviseagain{2.76} & \reviseagain{11.67} \\
  & & Err. std (\%)
  & \reviseagain{6.93} & \reviseagain{0.95} & \reviseagain{0.42} & \reviseagain{0.43} & \reviseagain{0.79} & \reviseagain{0.89} & \reviseagain{4.72} \\
  & & Time (s)
  & \reviseagain{61.37} & \reviseagain{107.08} & \reviseagain{107.01} & \reviseagain{69.84} & \reviseagain{155.28} & \reviseagain{1510.25} & \reviseagain{967.89} \\

\cmidrule(lr){2-10}
  & \multirow{3}{*}{\red{DLL-Histogram}} 
  & Rel. err (\%) 
  & \red{23.96} & \red{86.65} & \red{90.17} & \red{83.82} & \red{33.47} & \red{17.24} & \red{53.01} \\
  & & Err. std (\%) 
  & \orange{1.72} & \orange{28.45} & \orange{72.77} & \orange{79.19} & \orange{20.02} & \orange{13.23} & \orange{44.56} \\
  & & Time (s) 
  & \cellcolor{gray!12}\red{12.18} & \cellcolor{gray!12}\red{8.27} & \cellcolor{gray!12}\red{15.27} & \cellcolor{gray!12}\red{20.20} & 
  \cellcolor{gray!12}\red{21.41} & \cellcolor{gray!12}\red{33.34} & \cellcolor{gray!12}\red{66.74} \\
\bottomrule
\end{tabular}}
\caption{Mean relative error (\%), relative error std (\%), and runtime (s) for all methods across four query types on seven real-world datasets. 
The top row in each query block gives the true query result.}
\label{tab:task1}
\end{table*}

\label{subsec:experiment-result}
\paragraph{Utility and efficiency}

The experimental results, including relative errors with standard deviations and runtimes across all queries and datasets, are summarized in Table~\ref{tab:task1}.
\orange{N2E consistently achieves \reviseagain{competitive} utility in the majority of query-dataset combinations, considering both the mean and standard deviation of the relative errors.}



For $Q_{\mathrm{EC}}$, N2E achieves a relative error below $10\%$ for all datasets, and under $1\%$ for six out of seven datasets. 
Compared to \reviseagain{baselines OPT$^2$ and R2T}, N2E improves upon the error of OPT$^2$ by up to $2.5\times$, and R2T by up to \red{$7.5\times$}. 
In terms of runtime, N2E is faster than OPT$^2$ on six out of seven datasets and performs particularly well on graphs with smaller maximum degrees.
\reviseagain{
Moreover, Gem-EC achieves the lowest runtime among all methods while maintaining accuracy comparable to N2E.
  }

For $Q_{\mathrm{TP}}$, N2E is the only mechanism that completes on all datasets within the given time and memory constraints. 
OPT$^2$ fails to run on most datasets due to exceeding the time limit, while R2T runs out of memory on two datasets. 
Among the cases where baselines complete, N2E achieves comparable errors while reducing runtime by over $100\times$.
Notably, the error of N2E on the EUEmail dataset is relatively large due to clipping bias, which truncates the high contributions from nodes with very large degrees.

For $Q_{\mathrm{MD}}$, N2E maintains a relative rank error below $4\%$ across all datasets, demonstrating its robustness to variations in the maximum degree. 
\red{ 
N-Count provides no utility, exhibiting relative errors exceeding $3{,}000\%$ on all datasets.
}
The runtime of N2E for $Q_{\mathrm{MD}}$ is nearly identical to that for $Q_{\mathrm{EC}}$ and $Q_{\mathrm{TP}}$ across all datasets. 
This is expected, as all three queries follow the same steps of the N2E framework and involve an edge-DP mechanism with similar computational overhead.

For $Q_{\mathrm{DD}}$, N2E outperforms DLL-Histogram with up to \red{$80\times$} reduction in relative error across all datasets, albeit with up to \red{$8\times$} higher runtime. 
\reviseagain{Additionally, N2E achieves performance close to Gem-Histogram with slightly higher efficiency.}
The runtime of N2E on $Q_{\mathrm{DD}}$ is lower than that on $Q_{\mathrm{EC}}$, $Q_{\mathrm{TP}}$, and $Q_{\mathrm{MD}}$. 
This arises from the distinct clipping mechanism used. 
In the latter three queries, the distance-preserving clipping mechanism as defined in Algorithm~\ref{alg:clip_graph} evaluates edge rankings of both endpoints, introducing significant computational overhead for graphs with many edges. 
In contrast, the clipping mechanism in~\cite{day2016publishing} used for $Q_{\mathrm{DD}}$ considers only the degrees of the endpoints, resulting in reduced computational cost.

\paragraph{Consistency with theoretical guarantees}

As shown in Section~\ref{sec:app}, the absolute errors of $Q_{\mathrm{EC}}$ and $Q_{\mathrm{DD}}$ depend heavily on the maximum degree of the input graph, with higher maximum degrees leading to larger errors. 
Theoretical results of $Q_{\mathrm{TP}}$ also exhibit this dependency. 
However, since we use relative error for evaluation, the results are influenced by factors such as graph size and skewness, i.e., the uniformity of the degree distribution.
In general, datasets with larger graph sizes and less skewed degree distributions tend to exhibit lower relative errors.
For example, DBLP generally yields lower relative errors on the three queries due to its medium graph size and more balanced degree distribution.
For $Q_{\mathrm{MD}}$, the theoretical guarantee is provided in terms of edge distance.
Since we evaluate using relative rank error, additional factors beyond graph size and skewness also influence the results.
As a result, the performance trend remains consistent across the first three queries, but does not extend to $Q_{\mathrm{MD}}$.


Cluster-level connectivity appears to have minimal impact on query performance. 
For example, in $Q_{\mathrm{TP}}$, the Web-Google dataset achieves a relative error of $9.33\%$, nearly identical to the $10.71\%$ error observed on Amazon0601, despite a $10\times$ difference in connectivity.

\paragraph{Privacy parameter $\varepsilon$.}

\begin{figure*}[tbp]
    \centering
    \begin{subfigure}[b]{\textwidth}
        \centering
        \includegraphics[width=1\textwidth]{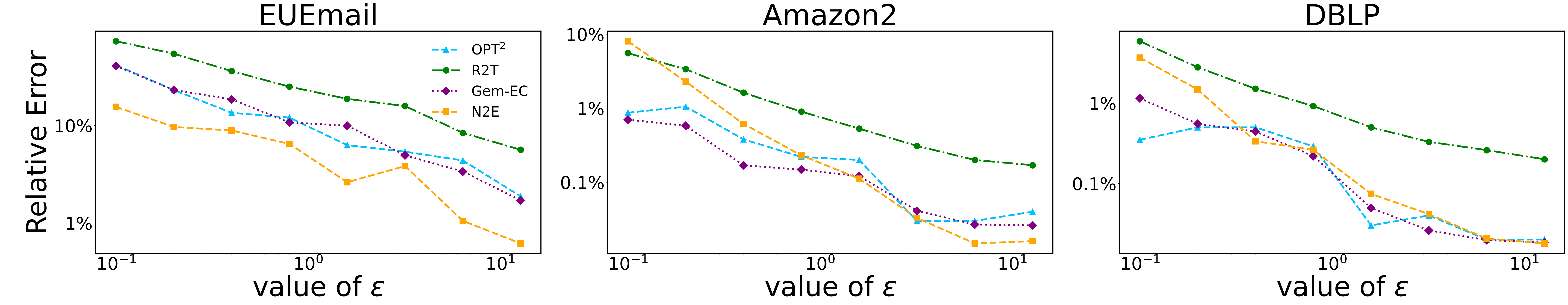}
        \caption{
        \reviseagain{Relative error (\%) for $Q_{\mathrm{EC}}$ on the EUEmail, Amazon2, and DBLP datasets across varying $\varepsilon$ values.}}
        \label{fig:sub1}
    \end{subfigure}

\vspace{1em}
    
    \begin{subfigure}[b]{\textwidth}
        \centering
        \includegraphics[width=1\textwidth]{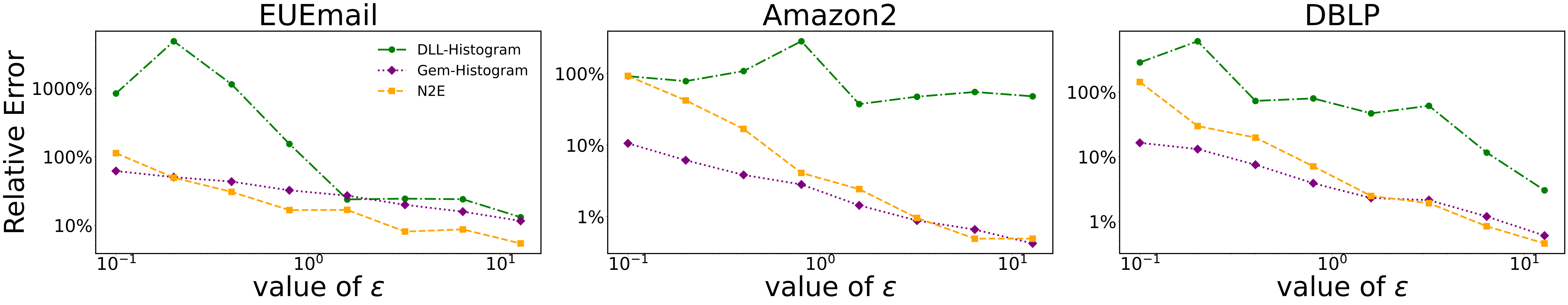}
        \caption{\reviseagain{Relative error (\%) for $Q_{\mathrm{DD}}$ on the EUEmail, Amazon2, and DBLP datasets across varying $\varepsilon$ values.}}
        \label{fig:sub3}
    \end{subfigure}
    
    \caption{Relative error (\%) of N2E and baseline methods for $Q_{\mathrm{EC}}$ and $Q_{\mathrm{DD}}$ on three datasets under varying privacy budgets $\varepsilon$.}
    \label{fig:task2}
\end{figure*}

To compare with the performance of N2E and baseline methods under different privacy guarantees, we vary the privacy budget $\varepsilon$ on a logarithmic scale from $0.1$ to $12.8$. 
Experiments are conducted on the EUEmail, Amazon2, and DBLP datasets to assess performance across two graph analytical tasks $Q_{\mathrm{EC}}$ and $Q_{\mathrm{DD}}$. 
The results are shown in Figure~\ref{fig:task2}. 
As shown, all methods demonstrate improved utility as $\varepsilon$ increases.
For $Q_{\mathrm{EC}}$, N2E outperforms R2T and achieves utility comparable to \reviseagain{the provably optimal baseline OPT$^2$ and Gem-EC.}
For $Q_{\mathrm{DD}}$,
\reviseagain{N2E generally outperforms DLL-Histogram and performs comparably to Gem-Histogram, which achieves better utility for small values of $\varepsilon$. 
}

Furthermore, N2E generally achieves higher utility on $Q_{\mathrm{EC}}$ than on $Q_{\mathrm{DD}}$ under small $\varepsilon$ values, due to the higher base value of the edge counting result. 
In contrast, the utility for $Q_{\mathrm{DD}}$ is more affected under small $\varepsilon$ because it involves answering multiple queries while having a smaller base value. 

\paragraph{Optimality of $\tau^*$.}

\begin{table}[htbp!]
\centering
\setlength{\tabcolsep}{6pt}
\resizebox{0.55\columnwidth}{!}{
\begin{tabular}{l||r|r|r}
\toprule
\textbf{Dataset} & \textbf{EUEmail} & \textbf{DBLP} & \textbf{Amazon0601} \\
\midrule\midrule
N2E (Rel. error (\%))     & 6.20     & 0.23    & 0.24  \\
\midrule\midrule
\(\tau^* = 1\)       & 95.44     & 99.83    & 99.61 \\
\(\tau^* = 4\)       & 93.14     & 85.08    & 92.30   \\
\(\tau^* = 16\)      & 88.40   & 35.98      & 30.74    \\
\(\tau^* = 64\)      & 75.26   & 4.38       & 6.22     \\
\(\tau^* = 256\)     & 44.49  & \cellcolor{gray!35}0.036 & 0.95     \\
\(\tau^* = 1024\)    & 11.93   & 0.22   & \cellcolor{gray!35}0.18    \\
\(\tau^* = 4096\)    & \cellcolor{gray!35}3.85 & 0.84     & 0.43   \\
\(\tau^* = 16384\)   & 14.42   & 1.86   & 1.16   \\
\bottomrule
\end{tabular}
}
\caption{Relative error (\%) of $Q_{\mathrm{EC}}$ on three datasets using N2E and different clipping thresholds \(\tau^*\).}
\label{tab:task3_ec}
\end{table}

\begin{table}[htbp!]
\centering
\setlength{\tabcolsep}{6pt}
\resizebox{0.55\columnwidth}{!}{
\begin{tabular}{l||r|r|r}
\toprule
\textbf{Dataset} & \textbf{EUEmail} & \textbf{DBLP} & \textbf{Amazon0601} \\
\midrule\midrule
N2E (Rel. error (\%))   & 8.17     & 1.94    & 3.75  \\
\midrule\midrule
\(\tau^* = 1\)       & 29.47     & 172.76 & 194.24 \\
\(\tau^* = 4\)       & 25.07     & 55.93   & 151.26   \\
\(\tau^* = 16\)      & 24.35   & 20.44   & 31.88    \\
\(\tau^* = 64\)      & 20.51   & 1.93 & 5.90     \\
\(\tau^* = 256\)     & 10.58   & \cellcolor{gray!35}0.37 & \cellcolor{gray!35}1.07     \\
\(\tau^* = 1024\)    & \cellcolor{gray!35}2.79   & 1.39 & 1.24    \\
\(\tau^* = 4096\)    & 9.02   & 6.76 & 6.28   \\
\(\tau^* = 16384\)   & 36.19   & 24.89 & 17.45   \\
\bottomrule
\end{tabular}
}
\caption{Relative error (\%) of $Q_{\mathrm{DD}}$ on three datasets using N2E and different clipping thresholds \(\tau^*\).}
\label{tab:task3_dd}
\end{table}

N2E relies on the maximum degree approximation $\tau^*$ as the degree upper bound to clip the input graph. 
To assess the performance of $\tau^*$ within N2E, we manually set different values of $\tau^*$ on a logarithmic scale as the clipping threshold, and evaluate the results on three datasets: EUEmail, DBLP, and Amazon0601.
The results for $Q_{\mathrm{EC}}$ and $Q_{\mathrm{DD}}$ are presented in Table~\ref{tab:task3_ec} and Table~\ref{tab:task3_dd}, respectively.

Compared to the optimal value of $\tau^*$, N2E incurs only a constant-factor utility loss in error. 
This is primarily due to the allocation of part of the privacy budget for computing $\tau^*$. 
Additionally, when comparing $Q_{\mathrm{EC}}$ and $Q_{\mathrm{DD}}$, the optimal value of $\tau^*$ can vary across queries on the same dataset. 
Despite this variation, the results show that the maximum degree approximation of N2E consistently delivers strong performance.


\section{Future Work}
\label{sec:futurework}

One direction for future exploration is extending the N2E framework to hypergraphs. Unlike conventional graphs, where there is at most one edge between any pair of nodes, hypergraphs allow hyperedges that can connect multiple nodes, posing new challenges for privacy preservation. 
Another promising direction is adapting N2E to decentralized settings, where the graph data is distributed across multiple parties, and each party has access only to an individual node and its adjacent edges. 

\section*{Acknowledgments}

This research is supported by the NTU–NAP Startup Grant (024584-00001) and the Singapore Ministry of Education Tier 1 Grant (RG19/25). 
\revise{We thank Adam Smith, Palak Jain, and Connor Wagaman for pointing out the overlap between their work~\cite{jain2024time} and ours on graph clipping for node-DP to edge-DP reduction, \reviseagain{as well as relevant baselines for edge counting and degree distribution that we had previously omitted~\cite{raskhodnikova2016lipschitz, kasiviswanathan2013analyzing}.}
Their feedback helped us improve the rigor and accuracy of the presentation.}
We would also like to thank the anonymous reviewers who have made valuable suggestions on improving the presentation of the paper.

\bibliographystyle{ACM-Reference-Format}
\bibliography{FULL_version}

@string{focs =  "Proc. IEEE Symposium on Foundations of Computer Science"}

@string{soda =  "Proc. ACM-SIAM Symposium on Discrete Algorithms"}

@string{esa  =  "Proc. European Symposium on Algorithms"}

@string{pods =  "Proc. ACM Symposium on Principles of Database Systems"}

@string{graph = "Proc. Graph-Theoretic Concepts in Computer Science"}

@string{vldb =  "Proc. International Conference on Very Large Data Bases"}

@string{icalp= "Proc.  International Colloquium on Automata, Languages,
                  and Programming"}

@string{sigmod = "Proc. ACM SIGMOD International Conference on Management of Data"}

@string{icdm = "Proc. IEEE International Conference on Data Mining"}

@string{computer = "IEEE Computer"}

@inproceedings{farhadi2022differentially,
  title={Differentially private densest subgraph},
  author={Farhadi, Alireza and Hajiaghayi, MohammadTaghi and Shi, Elaine},
  booktitle={International Conference on Artificial Intelligence and Statistics},
  pages={11581--11597},
  year={2022},
  organization={PMLR}
}

@article{liu2024pgb,
  title={PGB: Benchmarking Differentially Private Synthetic Graph Generation Algorithms},
  author={Liu, Shang and Du, Hao and Cao, Yang and Yan, Bo and Liu, Jinfei and Yoshikawa, Masatoshi},
  journal={arXiv preprint arXiv:2408.02928},
  year={2024}
}

@inproceedings{dinitz2025almost,
  title={Almost Tight Bounds for Differentially Private Densest Subgraph},
  author={Dinitz, Michael and Kale, Satyen and Lattanzi, Silvio and Vassilvitskii, Sergei},
  booktitle={Proceedings of the 2025 Annual ACM-SIAM Symposium on Discrete Algorithms (SODA)},
  pages={2908--2950},
  year={2025},
  organization={SIAM}
}

@inproceedings{deng2023differentially,
  title={Differentially private range query on shortest paths},
  author={Deng, Chengyuan and Gao, Jie and Upadhyay, Jalaj and Wang, Chen},
  booktitle={Algorithms and Data Structures Symposium},
  pages={340--370},
  year={2023},
  organization={Springer}
}

@inproceedings{chen2023differentially,
  title={Differentially private all-pairs shortest path distances: Improved algorithms and lower bounds},
  author={Chen, Justin Y and Ghazi, Badih and Kumar, Ravi and Manurangsi, Pasin and Narayanan, Shyam and Nelson, Jelani and Xu, Yinzhan},
  booktitle={Proceedings of the 2023 Annual ACM-SIAM Symposium on Discrete Algorithms (SODA)},
  pages={5040--5067},
  year={2023},
  organization={SIAM}
}

@article{blondel2008fast,
  title={Fast unfolding of communities in large networks},
  author={Blondel, Vincent D and Guillaume, Jean-Loup and Lambiotte, Renaud and Lefebvre, Etienne},
  journal={Journal of statistical mechanics: theory and experiment},
  volume={2008},
  number={10},
  pages={P10008},
  year={2008},
  publisher={IOP Publishing}
}

@inproceedings{nguyen2021differentially,
  title={Differentially private densest subgraph detection},
  author={Nguyen, Dung and Vullikanti, Anil},
  booktitle={International Conference on Machine Learning},
  pages={8140--8151},
  year={2021},
  organization={PMLR}
}

@article{manual1987ibm,
  title={Ibm ilog cplex optimization studio},
  author={Manual, CPLEX User’s},
  journal={Version},
  volume={12},
  number={1987-2018},
  pages={1},
  year={1987}
}

@inproceedings{cormode2012differentially,
  title={Differentially private summaries for sparse data},
  author={Cormode, Graham and Procopiuc, Cecilia and Srivastava, Divesh and Tran, Thanh TL},
  booktitle={Proceedings of the 15th International Conference on Database Theory},
  pages={299--311},
  year={2012}
}

@inproceedings{acs2012differentially,
  title={Differentially private histogram publishing through lossy compression},
  author={Acs, Gergely and Castelluccia, Claude and Chen, Rui},
  booktitle={2012 IEEE 12th International Conference on Data Mining},
  pages={1--10},
  year={2012},
  organization={IEEE}
}

@inproceedings{sealfon2016shortest,
  title={Shortest paths and distances with differential privacy},
  author={Sealfon, Adam},
  booktitle={Proceedings of the 35th ACM SIGMOD-SIGACT-SIGAI Symposium on Principles of Database Systems},
  pages={29--41},
  year={2016}
}

@article{fan2022private,
  title={Private graph all-pairwise-shortest-path distance release with improved error rate},
  author={Fan, Chenglin and Li, Ping and Li, Xiaoyun},
  journal={Advances in Neural Information Processing Systems},
  volume={35},
  pages={17844--17856},
  year={2022}
}

@inproceedings{hay2009accurate,
  title={Accurate estimation of the degree distribution of private networks},
  author={Hay, Michael and Li, Chao and Miklau, Gerome and Jensen, David},
  booktitle={2009 Ninth IEEE International Conference on Data Mining},
  pages={169--178},
  year={2009},
  organization={IEEE}
}

@inproceedings{sala2011sharing,
  title={Sharing graphs using differentially private graph models},
  author={Sala, Alessandra and Zhao, Xiaohan and Wilson, Christo and Zheng, Haitao and Zhao, Ben Y},
  booktitle={Proceedings of the 2011 ACM SIGCOMM conference on Internet measurement conference},
  pages={81--98},
  year={2011}
}

@inproceedings{narayanan2009anonymizing,
  title={De-anonymizing social networks},
  author={Narayanan, Arvind and Shmatikov, Vitaly},
  booktitle={2009 30th IEEE symposium on security and privacy},
  pages={173--187},
  year={2009},
  organization={IEEE}
}

@inproceedings{wu2021adapting,
  title={Adapting membership inference attacks to GNN for graph classification: Approaches and implications},
  author={Wu, Bang and Yang, Xiangwen and Pan, Shirui and Yuan, Xingliang},
  booktitle={2021 IEEE International Conference on Data Mining (ICDM)},
  pages={1421--1426},
  year={2021},
  organization={IEEE}
}

@inproceedings{zhang2022inference,
  title={Inference attacks against graph neural networks},
  author={Zhang, Zhikun and Chen, Min and Backes, Michael and Shen, Yun and Zhang, Yang},
  booktitle={31st USENIX Security Symposium (USENIX Security 22)},
  pages={4543--4560},
  year={2022}
}

@inproceedings{day2016publishing,
  title={Publishing graph degree distribution with node differential privacy},
  author={Day, Wei-Yen and Li, Ninghui and Lyu, Min},
  booktitle={Proceedings of the 2016 International Conference on Management of Data},
  pages={123--138},
  year={2016}
}

@inproceedings{zhang2014towards,
  title={Towards accurate histogram publication under differential privacy},
  author={Zhang, Xiaojian and Chen, Rui and Xu, Jianliang and Meng, Xiaofeng and Xie, Yingtao},
  booktitle={Proceedings of the 2014 SIAM international conference on data mining},
  pages={587--595},
  year={2014},
  organization={SIAM}
}

@article{qardaji2013understanding,
  title={Understanding hierarchical methods for differentially private histograms},
  author={Qardaji, Wahbeh and Yang, Weining and Li, Ninghui},
  journal={Proceedings of the VLDB Endowment},
  volume={6},
  number={14},
  pages={1954--1965},
  year={2013},
  publisher={VLDB Endowment}
}

@inproceedings{blocki2013differentially,
  title={Differentially private data analysis of social networks via restricted sensitivity},
  author={Blocki, Jeremiah and Blum, Avrim and Datta, Anupam and Sheffet, Or},
  booktitle={Proceedings of the 4th conference on Innovations in Theoretical Computer Science},
  pages={87--96},
  year={2013}
}

@article{dong2023better,
  title={Better than Composition: How to Answer Multiple Relational Queries under Differential Privacy},
  author={Dong, Wei and Sun, Dajun and Yi, Ke},
  journal={Proceedings of the ACM on Management of Data},
  volume={1},
  number={2},
  pages={1--26},
  year={2023},
  publisher={ACM New York, NY, USA}
}

@inproceedings{kasiviswanathan2013analyzing,
  title={Analyzing graphs with node differential privacy},
  author={Kasiviswanathan, Shiva Prasad and Nissim, Kobbi and Raskhodnikova, Sofya and Smith, Adam},
  booktitle={Theory of Cryptography Conference},
  pages={457--476},
  year={2013},
  organization={Springer}
}

@article{karwa2011private,
  title={Private analysis of graph structure},
  author={Karwa, Vishesh and Raskhodnikova, Sofya and Smith, Adam and Yaroslavtsev, Grigory},
  journal={Proceedings of the VLDB Endowment},
  volume={4},
  number={11},
  pages={1146--1157},
  year={2011},
  publisher={VLDB Endowment}
}

@inproceedings{chen2013recursive,
  title={Recursive mechanism: towards node differential privacy and unrestricted joins},
  author={Chen, Shixi and Zhou, Shuigeng},
  booktitle={Proceedings of the 2013 ACM SIGMOD International Conference on Management of Data},
  pages={653--664},
  year={2013}
}

@article{johnson2018towards,
  title={Towards practical differential privacy for SQL queries},
  author={Johnson, Noah and Near, Joseph P and Song, Dawn},
  journal={Proceedings of the VLDB Endowment},
  volume={11},
  number={5},
  pages={526--539},
  year={2018},
  publisher={VLDB Endowment}
}

@inproceedings{dwork2006calibrating,
  title={Calibrating noise to sensitivity in private data analysis},
  author={Dwork, Cynthia and McSherry, Frank and Nissim, Kobbi and Smith, Adam},
  booktitle={Theory of cryptography conference},
  pages={265--284},
  year={2006},
  organization={Springer}
}

@inproceedings{liu2019private,
  title={Private selection from private candidates},
  author={Liu, Jingcheng and Talwar, Kunal},
  booktitle={Proceedings of the 51st Annual ACM SIGACT Symposium on Theory of Computing},
  pages={298--309},
  year={2019}
}

@inproceedings{raskhodnikova2016lipschitz,
  title={Lipschitz extensions for node-private graph statistics and the generalized exponential mechanism},
  author={Raskhodnikova, Sofya and Smith, Adam},
  booktitle={2016 IEEE 57th Annual Symposium on Foundations of Computer Science (FOCS)},
  pages={495--504},
  year={2016},
  organization={IEEE}
}

@inproceedings{nissim2007smooth,
  title={Smooth sensitivity and sampling in private data analysis},
  author={Nissim, Kobbi and Raskhodnikova, Sofya and Smith, Adam},
  booktitle={Proceedings of the thirty-ninth annual ACM symposium on Theory of computing},
  pages={75--84},
  year={2007}
}

@inproceedings{zhang2015private,
  title={Private release of graph statistics using ladder functions},
  author={Zhang, Jun and Cormode, Graham and Procopiuc, Cecilia M and Srivastava, Divesh and Xiao, Xiaokui},
  booktitle={Proceedings of the 2015 ACM SIGMOD international conference on management of data},
  pages={731--745},
  year={2015}
}

@article{dwork2014algorithmic,
  title={The algorithmic foundations of differential privacy},
  author={Dwork, Cynthia and Roth, Aaron},
  journal={Foundations and Trends{\textregistered} in Theoretical Computer Science},
  volume={9},
  number={3--4},
  pages={211--407},
  year={2014},
  publisher={Now Publishers, Inc.}
}

@inproceedings{Imola2021,
   title = {Locally Differentially Private Analysis of Graph Statistics},
   author = {Jacob Imola and Takao Murakami and Kamalika Chaudhuri},
   booktitle = {30th USENIX Security Symposium (USENIX Security 21)},
   pages = {983--1000},
   year = {2021},
   publisher = {USENIX Association},
}

@inproceedings {Imola2022,
   title = {Communication-Efficient Triangle Counting under Local Differential Privacy},
   author = {Jacob Imola and Takao Murakami and Kamalika Chaudhuri},
   booktitle = {31st USENIX Security Symposium (USENIX Security 22)},
   pages = {537--554},
   year = {2022},
   publisher = {USENIX Association},
   address = {Boston, MA},
}

@inproceedings{He2025,
   title = {Robust Privacy-Preserving Triangle Counting under Edge Local Differential Privacy},
   author = {Yizhang He and Kai Wang and Wenjie Zhang and Xuemin Lin and Ying Zhang and Wei Ni},
   booktitle = {Proc. ACM SIGMOD International Conference on Management of Data},
   year = {2025},
   publisher = {Association for Computing Machinery},
   address = {New York, NY, USA},
}

@inproceedings{eden2023triangle,
  title={Triangle Counting with Local Edge Differential Privacy},
  author={Eden, Talya and Liu, Quanquan C and Raskhodnikova, Sofya and Smith, Adam},
  booktitle={50th International Colloquium on Automata, Languages, and Programming (ICALP 2023)},
  pages={52--1},
  year={2023},
  organization={Schloss Dagstuhl--Leibniz-Zentrum f{\"u}r Informatik}
}

@article{dong2024instance,
  title={Instance-optimal Truncation for Differentially Private Query Evaluation with Foreign Keys},
  author={Dong, Wei and Fang, Juanru and Yi, Ke and Tao, Yuchao and Machanavajjhala, Ashwin},
  journal={ACM Transactions on Database Systems},
  volume={49},
  number={4},
  pages={1--40},
  year={2024},
  publisher={ACM New York, NY}
}

@inproceedings{dong21:residual,
  title={Residual Sensitivity for Differentially Private Multi-Way Joins},
  author={Wei Dong and Ke Yi},
   booktitle=sigmod,
  year={2021}
}

@inproceedings{jain2024time,
  title={Time-aware projections: Truly node-private graph statistics under continual observation},
  author={Jain, Palak and Smith, Adam and Wagaman, Connor},
  booktitle={2024 IEEE Symposium on Security and Privacy (SP)},
  pages={127--145},
  year={2024},
  organization={IEEE}
}

@inproceedings{zhang2020community,
  title={Community preserved social graph publishing with node differential privacy},
  author={Zhang, Sen and Ni, Weiwei and Fu, Nan},
  booktitle={2020 IEEE International Conference on Data Mining (ICDM)},
  pages={1400--1405},
  year={2020},
  organization={IEEE}
}

@misc{leskovec2014snap,
  title={SNAP: Stanford network analysis project},
  author={Leskovec, Jure and Krevl, Andrej},
  journal={2014-02-01). https://snap. stanford. edu},
  year={2014}
}

@inproceedings{huang21mean,
  title={Instance-optimal Mean Estimation Under Differential Privacy},
  author={Ziyue Huang and Yuting Liang and Ke Yi},
  booktitle={NeurIPS},
  year={2021}
}

@inproceedings{dong2021nearly,
  title={A Nearly Instance-optimal Differentially Private Mechanism for Conjunctive Queries},
  author={Wei Dong and Ke Yi},
  booktitle=pods,
  year={2022}
}

@inproceedings{dong2022r2t,
  title={R2T: Instance-optimal Truncation for Differentially Private Query Evaluation with Foreign Keys},
  author={Dong, Wei and Fang, Juanru and Yi, Ke and Tao, Yuchao and Machanavajjhala, Ashwin},
  booktitle={Proc. ACM SIGMOD International Conference on Management of Data},
  year={2022}
}

@inproceedings{fang2022shifted,
  title={Shifted Inverse: A General Mechanism for Monotonic Functions under User Differential Privacy},
  author={Fang, Juanru and Dong, Wei and Yi, Ke},
  booktitle={Proceedings of the 2022 ACM SIGSAC Conference on Computer and Communications Security},
  year={2022}
}

@article{liu2024unleash,
  title={Unleash the Power of Ellipsis: Accuracy-enhanced Sparse Vector Technique with Exponential Noise},
  author={Liu, Yuhan and Wang, Sheng and Liu, Yixuan and Li, Feifei and Chen, Hong},
  journal={arXiv preprint arXiv:2407.20068},
  year={2024}
}

@inproceedings{dong2023universal,
  title={Universal private estimators},
  author={Dong, Wei and Yi, Ke},
  booktitle={Proceedings of the 42nd ACM SIGMOD-SIGACT-SIGAI Symposium on Principles of Database Systems},
  pages={195--206},
  year={2023}
}

@inproceedings{sajadmanesh2023gap,
  title={$\{$GAP$\}$: Differentially private graph neural networks with aggregation perturbation},
  author={Sajadmanesh, Sina and Shamsabadi, Ali Shahin and Bellet, Aur{\'e}lien and Gatica-Perez, Daniel},
  booktitle={32nd USENIX Security Symposium (USENIX Security 23)},
  pages={3223--3240},
  year={2023}
}

@inproceedings{dwork2009complexity,
  title={On the complexity of differentially private data release: efficient algorithms and hardness results},
  author={Dwork, Cynthia and Naor, Moni and Reingold, Omer and Rothblum, Guy N and Vadhan, Salil},
  booktitle={Proceedings of the forty-first annual ACM symposium on Theory of computing},
  pages={381--390},
  year={2009}
}

@inproceedings{fichtenberger2021differentially,
  title={Differentially Private Algorithms for Graphs Under Continual Observation},
  author={Fichtenberger, Hendrik and Henzinger, Monika and Ost, Wolfgang},
  booktitle={29th Annual European Symposium on Algorithms (ESA 2021)},
  year={2021},
  organization={Schloss Dagstuhl-Leibniz-Zentrum f{\"u}r Informatik}
}

@article{dong2024continual,
  title={Continual observation of joins under differential privacy},
  author={Dong, Wei and Chen, Zijun and Luo, Qiyao and Shi, Elaine and Yi, Ke},
  journal={Proceedings of the ACM on Management of Data},
  volume={2},
  number={3},
  pages={1--27},
  year={2024},
  publisher={ACM New York, NY, USA}
}

@article{full_version,
  title={{N2E}: A General Framework to Reduce Node-Differential Privacy to Edge-Differential Privacy for Graph Analytics [Full Version]},
  author={Hu, Yihua and Ding, Hao and Dong, Wei},
  url={https://drive.google.com/drive/folders/1T9SLrU9z406RdMaWKPg3KP6aWEJOo73B?usp=sharing},
  year={2025}
}


\end{document}